\def\draft{0}
\def\llncs{0}
\def\anon{0}

\ifnum\llncs=1
\documentclass[envcountsect,envcountsame,orivec,runningheads]{llncs}
\usepackage{amssymb,url}
\usepackage[noadjust]{cite}
\else
\documentclass[letterpaper]{article}
\usepackage{amsfonts,amsmath,amssymb,amsthm}
\usepackage{fullpage}
\usepackage{palatino}
\fi
\usepackage{physics}
\usepackage{pifont}
\usepackage[english]{babel}
\usepackage[T1]{fontenc}
\usepackage[utf8]{inputenc}

\usepackage{tikz}
\usetikzlibrary{quantikz}

\usepackage[colorlinks=true,linkcolor=magenta,citecolor=blue]{hyperref}
\usepackage[nameinlink]{cleveref}

\usepackage{comment}
\usepackage{lscape}
\usepackage{mathtools,xparse}

\ifnum\llncs=1
\title{A New Approach to Generic Lower Bounds: 
\large Classical/Quantum MDL, Quantum Factoring, and More}
\ifnum\anon=0
\author{%
Minki Hhan\inst{1}
}
\institute{
Institute 1\\
\email{email 1}\\
}
\fi
\else
\theoremstyle{plain}
\title{A New Approach to Generic Lower Bounds: \\
\large Classical/Quantum MDL, Quantum Factoring, and More}
\ifnum\anon=0
\author{
Minki Hhan\thanks{\texttt{E-mail:minkihhan@gmail.com}. KIAS}}
\fi
\fi

\newcommand{\E}{\mathbb E}

\newcommand{\cA}{\mathcal{A}}
\newcommand{\cB}{\mathcal{B}}
\newcommand{\cC}{\mathcal{C}}
\newcommand{\cM}{\mathcal{M}}

\newcommand{\cD}{\mathcal{D}}

\newcommand{\cP}{\mathcal{P}}
\newcommand{\cG}{\mathcal{G}}

\newcommand{\cZ}{\mathcal{Z}}

\newcommand{\cR}{\mathcal{R}}
\newcommand{\Z}{\mathbb{Z}}

\newcommand{\veca}{\mathbf a}
\newcommand{\vecb}{\mathbf b}
\newcommand{\vecc}{\mathbf c}

\newcommand{\vecw}{\mathbf w}
\newcommand{\vecx}{\mathbf x}

\newcommand{\vecz}{\mathbf z}

\newcommand{\vecB}{\mathbf B}
\newcommand{\vecC}{\mathbf C}

\newcommand{\vecX}{\mathbf X}
\newcommand{\vecY}{\mathbf Y}
\newcommand{\vecZ}{\mathbf Z}
\newcommand{\vecT}{\mathbf T}
\newcommand{\vecW}{\mathbf W}

\newcommand{\bit}{\{0,1\}}

\ifnum\llncs=0
  \theoremstyle{plain}% default
  \newtheorem{theorem}{Theorem}[section]
  \newtheorem{lemma}[theorem]{Lemma}
  \newtheorem{corollary}[theorem]{Corollary}
  
  \theoremstyle{definition}
  \newtheorem{remark}{Remark}
  \newtheorem{problem}{Problem}%[section]
\fi

\newtheorem{fact}{Fact}
\providecommand{\qedhere}{
\ifmmode
  \eqno \def\@badmath{$$}%$$
    \let\eqno\relax \let\leqno\relax \let\veqno\relax
    \hbox{\qed}
\else
  \qed
\fi
}

\newcommand{\dl}{{\sf{DL}}}
\newcommand{\chal}{{\sf{Chal}}}
\newcommand{\omdl}{{\sf{OM\text{-}DL}}}
\newcommand{\mmdl}{{\sf{\text{-}M\text{-}DL}}}
\newcommand{\gapdl}{{\sf{Gap\text{-}DL}}}
\newcommand{\gapcdh}{{\sf{Gap\text{-}CDH}}}
\newcommand{\ddh}{{\sf{DDH}}}
\newcommand{\mdl}{{\sf{MDL}}}
\newcommand{\safe}{{\sf{safe}}}
\newcommand{\pprime}{{\sf{prime}}}
\newcommand{\ord}{{\sf{ord}}}

\newcommand{\base}{{\sf{base}}}
\newcommand{\smooth}{{\sf{smooth}}}

\newcommand{\spann}{{\sf{span}}}

\ifnum\draft=1
\newcommand{\han}[1]{\textcolor{blue}{ $\langle \! \langle$ Minki: ``#1" $\rangle \! \rangle$}}
\else 
\newcommand{\han}[1]{}
\fi

\newcommand{\Gop}{{\cG.{\sf {op}}}}
\newcommand{\Geq}{{\cG.{\sf {eq}}}}
\newcommand{\Ginv}{{\cG.{\sf {inv}}}}

\newcommand{\FindLabel}{{{\sf {FindLabel}}}}
\newcommand{\FindElt}{{{\sf {FindElement}}}}

\newcommand{\Encode}{{{\sf {Encode}}}}
\newcommand{\Decode}{{{\sf {Decode}}}}
\newcommand{\eq}{{{\sf {eq}}}}
\newcommand{\Delegate}{{\sf{Delegate}}}
\newcommand{\DGop}{{\Delegate.{\sf {Gop}}}}
\newcommand{\DGinv}{{\Delegate.{\sf {Ginv}}}}
\newcommand{\DGeq}{{\Delegate.{\sf {Geq}}}}
\newcommand{\Test}{{{\sf {Test}}}}

\begin{document}
\maketitle
\begin{abstract}
This paper studies the limitations of the generic approaches to solving cryptographic problems in classical and quantum settings in various models. 
\begin{itemize}
    \item In the classical generic group model (GGM), we find simple alternative proofs for the lower bounds of variants of the discrete logarithm (DL) problem: the multiple-instance DL and one-more DL problems (and their mixture). We also re-prove the unknown-order GGM lower bounds, such as the order finding, root extraction, and repeated squaring.
    \item In the quantum generic group model (QGGM), 
    we study the complexity of variants of the discrete logarithm. We prove the logarithm DL lower bound in the QGGM even for the composite order setting. We also prove an asymptotically tight lower bound for the multiple-instance DL problem. Both results resolve the open problems suggested in a recent work by Hhan, Yamakawa, and Yun.
    \item In the quantum generic ring model we newly suggested,
    we give the logarithmic lower bound for the order-finding algorithms, an important step for Shor's algorithm.
    We also give a logarithmic lower bound for a certain generic factoring algorithm outputting relatively small integers, which includes a modified version of Regev's algorithm.
    \item Finally, we prove a lower bound for the basic index calculus method for solving the DL problem in a new idealized group model regarding smooth numbers.
\end{itemize}
The quantum lower bounds in both models allow certain (different) types of classical preprocessing.

All of the proofs are significantly simpler than the previous proofs and are through a single tool, the so-called compression lemma, along with linear algebra tools. Our use of this lemma may be of independent interest.
\end{abstract}
\ifnum\llncs=0 
\newpage
\tableofcontents
\newpage
\fi
\section{Introduction}
What is the source of the generic hardness of some cryptographic problems?

The generic group models (GGM)~\cite{Nec94,Shoup97,Maurer05} are the most successful and influential idealized models in cryptography. In this model, the group operations can be carried out by making queries to a group oracle, and any other use of the particular features of the group is not allowed.
Despite its restricted nature, many important algorithms, such as Pohlig-Hellman~\cite{PH78} or Pollard's rho algorithm~\cite{Pol78}, are encompassed by the class of generic group algorithms. Despite some criticisms~\cite{Dent02,KM06} and non-generic algorithms, e.g., index-calculus, the GGM plays an important test bed for cryptographic protocols, and the security proofs in the GGM provide a sanity check guaranteeing that there are no simple attacks. The proofs in the GGM become more meaningful in the elliptic-curve groups.

The GGM is especially promising because of its simple security proofs; most of the security proofs in the GGM heavily rely on the Schwartz-Zippel (SZ) lemma that is already used in~\cite{Shoup97}. This lemma roughly states that for a non-zero multivariate linear polynomial $P$ over $\Z_p$ for a prime $p$, the probability that a random element becomes a root of $P$ is $1/p$. This provides a meaningful limitation for generic algorithms obtaining a single piece of information and is used to prove the generic lower bounds for the discrete logarithm (DL) problem and computational/decisional Diffie-Hellman (C/DDH) problems, as well as many cryptographic applications.

We face hurdles in proving the generic security when we slightly tweak the model or problems. If we consider the unknown-order groups, the lower bounds for various problems can be proven with relatively small efforts, including the order-finding problem~\cite{Sut07} or the root extraction problems~\cite{DK02}. When we consider the problems where enormous amounts of information can be obtained, the security proofs based on the SZ lemma do not work. To remedy this, other idealized problems (e.g., the search-by-hyperplane/surfaces)~\cite{Yun15,AGK20,AHP23} are suggested or seemingly involved techniques (e.g., compression lemmas or pre-sampling) are used~\cite{CK18,CDG18} in the proofs. In the quantum setting, the rigorous proofs are rather complicated and pass through the classical lower bounds~\cite{HYY23}.

Extensions beyond the group structures~\cite{BL96,BR98,AM09,JS13,YYHK20} become much more complicated. In many cases, there is some evidence that the unconditional lower bounds unlikely exist, and the proofs are done through the reduction between the problems. To our knowledge, there are no known unconditional lower bounds, even in the idealized models. 

This state of affairs makes the genuine source of the generic hardness elusive and asks for case-by-case studies for each model. In particular, the unconditional lower bounds in idealized settings are only known for the generic group models.

\subsection{Our Results}
We provide a unified way to prove the old and new hardness proofs in the various idealized models: the known/unknown-order classical generic groups, quantum generic groups, quantum generic \emph{rings}, and the new group model embracing index calculus.

Our main technical lemma is, along with some linear algebraic observations, (variants of) \emph{the compression lemma}, which roughly asserts that there is no way to compress $n$-bit strings to strings less than $n$-bit.
This lemma is occasionally used in proving the time-space tradeoff lower bounds~\cite{DTT10,NABT15,DGK17,CK18,HXY19,CLL19}, and introduces a highly involved proof.
Our proofs are significantly simpler, as shown in this section. Roughly, we compress the problem instances along with some relevant randomness into the information that generic algorithms can obtain; the decoding simulates the generic algorithm using the encoding, \emph{without} accessing the oracles, but still recovers the problem instances. This gives some clues that the generic hardness is from the limited way of obtaining information on generic algorithms.

This paper mainly focuses on the abstract model of Maurer~\cite{Maurer05}. In this model, the group (or ring) elements are stored in \emph{element} wires, and they can be accessed only by the group operation gates or the equality gates.

\paragraph{The Known-order GGM Lower Bounds}
Let us begin with the lower bound of folklore for the DL problem (\Cref{thm: GGM_DL}). 
Let $\cG\simeq \Z_p$ be the underlying group of prime order $p$.
In this problem, the algorithm $\cA$ is given $(g,g^x)$ and is asked to find $x$.
Suppose that $\cA$ solves the DL problem with almost certainty with $T$ group operations. 
We also associate a polynomial $aX+b$ to the group element $g^{ax+b}$.
It is not hard to argue that (slightly modified) $\cA$ finds two equal group elements with different polynomials.

We use this algorithm to compress $x\in[p]$. Among $T$ group elements, there are $T^2$ possibilities for a pair of equal group elements. In other words, we can encode the discrete logarithm $x$ in $T^2$ possible collisions, and the compression lemma says that $\log T^2 = 2\log T \ge \log p$ for the high success probability. This implies that $T \ge \sqrt{|\cG|}$ as in the previous proofs.

We proceed to the MDL problem (\Cref{thm:MDL_GGM}). Suppose that the algorithm $\cA$ with $T$ group operations is given $(g,g^{x_1},...,g^{x_m})$ and is asked to find $\vecx=(x_1,...,x_m)$. As before, we can assume that $\cA$ finds $m$ collisions. We can encode $\vecx$ using the information of the collisions, which requires
\[
    \log \binom{\binom{T}{2}}{m} \approx m \log \frac{eT^2}{2m}
\]
bits. To solve the MDL problem with certainty, it must be larger than $m \log |\cG|$, which is the information that $\vecx$ possesses, implying that $T \ge \sqrt{mp}$. Previously, this bound was first proven in~\cite{Yun15} and required an involved argument regarding the related problem called the search-by-hyperplane-queries (SHQ). We note that we have \emph{another} simple proof for this lower bound solely based on the linear-algebra reasoning in~\Cref{thm:APP_MDL}.

The same proof strategy easily extends to the other problems. This includes the gap-DL and gap-CDH problems (\Cref{thm: GGM_GapDL,thm: GGM_GapCDH}) and the one-more DL problem (OM-DL) (\Cref{thm: GGM_OMDL}) that was first proven recently~\cite{BFP21} (and was falsely proven in~\cite{CDG18}). We actually prove the lower bounds for a much more general problem, where the adversary is asked to find the $n$-more DL solutions than its queries to the DL oracles; $n=1$ corresponds to the OM-DL problem.

\paragraph{The Unknown-order GGM Lower Bounds}
We also consider the unknown-order GGM. In~\Cref{sec: unknownGGM}, we show that the same strategy can prove the lower bounds for the order-finding in the prime-order group (\Cref{thm: uGGM-order}) that is shown in~\cite{Sut07} and in the RSA group (\Cref{thm: uGGM-RSA}). We also prove the hardness of the root extraction (\Cref{thm: uGGM_rootext}), which was proven in~\cite{DK02}, and the repeated squaring (\Cref{thm: uGGM_repeatsq}) in the unknown-order GGM. 
We stress that we do \emph{not} consider the ring operations. Thus, its implications are limited to the group setting.

We sketch the proof for the order-finding problem. In this model, the generic algorithm can compute $g^{x\pm y}$ for given $g^x,g^y$ as in the previous GGM, but does not know the order of the underlying group. Therefore, the corresponding polynomials have a bounded coefficient after $T$ group operations, so the number of their prime factors is bounded. It turns out that each equality gate can contain $T$ prime divisors. The encoding contains the equality gate that specifies the order, and the index of its divisors. The length becomes $3\log T$ to compress $\log |\cG|$-bit order, giving the $T \ge |\cG|^{1/3}$ bound.

\paragraph{The Quantum GGM Lower Bounds}
We prove the quantum lower bounds for solving the DL problem (\Cref{thm: QGGM_DL}) and variants in the quantum GGM (QGGM). This direction was suggested in~\cite{HYY23}, and the authors gave the lower bounds for the DL and C/DDH problems in the QGGM.

The proof strategy is different from the classical lower bounds. Instead of the one-shot encoding as in the classical setting, we need an interactive version of the compression lemma (\Cref{lem: interactive compression}) proven in~\cite{HNR18}. This roughly states that if Alice wants to send an $n$-bit message to Bob, Alice needs to send $n$-bit anyway, regardless of the amounts of Bob's messages to Alice and the number of rounds.

In the QGGM, the algorithm can make group operations coherently. Given a generic algorithm, we construct the interactive protocol between Alice and Bob, where Alice holds all the group elements, and Bob holds the other registers. Bob runs the DL algorithm, and whenever it needs to make a quantum group operation, he sends the relevant registers to Alice; Alice applies the group operations and returns the relevant registers to Bob. 
For simplicity, we assume that the indices for the target group elements are classical. In this setting, Alice sends two bits (or one qubit) to delegate the group operation Bob requested.
If the algorithm makes $Q$ group operations, the interactive version of the compression lemmas proves $2Q \ge \log |\cG|$, recovering the previous lower bound; actually, with a better constant than the previous bound $4Q\ge \log |\cG|$.

If we allow the indices to be quantum, delegating quantum group operations requires more communication to include them. 
The lower bound becomes $Q = \Omega(\log |\cG|/\log\ell)$ for the length of indices $\ell$.
The same strategy naturally extends to the MDL problem (\Cref{thm: QGGM_fine-grained-DL}) in the QGGM, proving $Q =\Omega(m\log G/\log \ell.)$ 

Our proof equally works for the \emph{composite order} DL problems and holds even regarding the classical preprocessing.
The composite order DL lower bound and MDL lower bound in the QGGM resolves the open problems asked in~\cite{HYY23}, where the matching algorithms were suggested.
In fact, our lower bound implies that the number of quantumly accessible indices is an important measure, while the previous results only consider the memory-bounded setting, which naturally bounds the number of quantum indices.
Also, our lower bound implies that the speed-up for the MDL problem beyond Shor in terms of the group operation complexity requires a large quantum data structure.

We also prove the QGGM variant for the order-finding problems (\Cref{thm: uQGGM-order}), showing the order-finding in the QGGM requires $\Omega(\log |\cG|)$ quantum group operations, even with classical preprocessing.

\paragraph{The Quantum Generic Ring Model and Lower Bounds}
We study a quantum variant of the generic ring model~\cite{AM09,JS13}, which we call the quantum generic ring model (QGRM). 
In this model, the algorithm has oracle access to the ring elements as in the GGM.
However, we do \emph{not} give the explicit value of $N$ to the algorithm because we aim for the unconditional lower bounds in the idealized model. If the algorithm knows $N$, we cannot rule out the direct use of $N$, and the proof must be through reductions as in~\cite{AM09}.

We prove that the logarithmic lower bound for the QGRM order finding algorithm in the ring isomorphic to $\Z_N$ where $N$ is a product of two safe primes (\Cref{thm: QGRM_order}). The order (or period) finding problem is a major subroutine in Shor's factoring algorithm~\cite{Shor99}. 

Note that a recent work of Regev~\cite{Regev23} solves the integer factorization with a different method. In this approach, small integers are extensively used, taking advantage of the fact that small integer arithmetic operations are faster than large integer operations, giving an improved algorithm with better circuit complexity.

We observe that this advantage results in a modified algorithm that outputs a plain integer with a nontrivial common factor with $N$ of relatively small size. We consider the generic algorithms that output such an integer to solve the integer factoring that can be computed \emph{without} modulus reductions---this must be done in plain because QGRM algorithms do not know $N$. We prove that if the output is relatively small, the logarithmic ring operation lower bound holds for factoring (\Cref{thm: QGRM_fact}). Intriguingly, the output of Shor's algorithm with this modification is too large to apply this lower bound.

These results give the first evidence that the quantum factoring algorithm needs a logarithmic number of group operations. Our result extends to the straight-line classical preprocessing that reflects the real world better.

\paragraph{Beyond GGMs: Index Calculus}
Finally, we study the idealized group model, called the smooth GGM, beyond the generic groups, encompassing the index calculus method.
This model provides the abstraction for the notion of smooth elements and efficient factoring for the smooth integers. 

We prove that the DL algorithm must make $\exp(C\sqrt{\log |\cG|\log\log |\cG|} )$ group operations for some constant $C>0$ in the SGGM (\Cref{thm: SGGM-DL}), giving some evidence that going beyond this bound requires a new idea, as the ones in the number field sieves.

We do not claim this lower bound provides new insights or {strong} evidence for the index calculus. We believe that the ideas used in the proof for the SGGM lower bound must have been observed and used in the development of the index calculus, especially for optimization.
Still, our result shows that a proper abstraction of the generic approaches, where only limited operations are used, can indeed prove that these approaches cannot go further; asking for new ideas.

\paragraph{Notations.}
For a positive integer $N$, a finite cyclic group of order $N$ is denoted by $\Z_N$, identified by $\{0,1,..., N-1\}$ with the natural group operation, and $[N]:=\{1,...,N\}$.
\section{Compression Lemmas}
This section presents our main lemmas, which are usually called the compression lemma.
The classical compression lemma is stated as follows.
\begin{lemma}\label{lem: compression}
    Let $\cM,R$ be finite sets. Let $\Encode:\cM\times R \rightarrow \bit^m$ and $\Decode:\bit^m\times R \rightarrow \cM$ be deterministic algorithms. For $\epsilon \in (0,1]$, if
    \[
        \Pr_{r \gets R,x \gets \cM} \left[
            \Decode(\Encode(x,r),r) = x
        \right] \ge \epsilon,
    \]
    then we have $m \ge \log |\cM| + \log \epsilon$.
\end{lemma}
This is a direct corollary of the following quantum interactive version of the compression lemma. Precisely, the classical one-way protocol with the preshared entanglement $\sum_{r \in R}\ket{r,r}$ corresponds to the above lemma.

\begin{lemma}[{\cite[Theorem 1.2]{HNR18}}]
    Consider an interactive protocol between Alice and Bob, who share an arbitrarily entangled state and communicate through classical channels.
    Alice wants to send a uniformly random element in a finite set $\cM$ to Bob.
    Suppose that the probability that Bob correctly recovers $x$ with probability $\epsilon\in (0,1]$, and Alice sends $m$ bits to Bob total over all rounds. Then it holds that $m \ge  \log |\cM|+ \log {\epsilon}$, regardless of the number of bits sent by Bob to Alice.

    In general, if Alice sends a classical string in $[M_i]$ to Bob as the $i$-th round message for $i\in [k]$ where $k$ is the maximum number of rounds, then it holds that
    \[\log \left(\prod_{i=1}^k M_i\right) \ge  \log |\cM|+ \log {\epsilon}.\] 
\end{lemma}

The original theorem in \cite[Theorem 1.2]{HNR18} mainly concerns the case of $\cM=\bit^n$ and the bit-strings as messages. This generalization is straightforward.\footnote{Roughly, the choice of $\cM=\bit^n$ is only used at the end of the proof where the probability that input to Alice is $x$ is $1/2^n$, and modifying it to $1/|\cM|$ suffices to prove our theorem. The non-bit-string is slightly involved, but changing the appropriate set suffices.}
The quantum communication version can be derived using quantum teleportation (See also~\cite[Theorem 2]{NS06}).

\begin{corollary}\label{lem: interactive compression}
    In the same setting as the above lemma, if Alice and Bob can communicate through \emph{quantum channels}, the bounds become
    \[m \ge  \frac{\log |\cM|+ \log {\epsilon}}2,~~~\text{ and }~~\log \left(\prod_{i=1}^k M_i\right) \ge \frac{ \log |\cM|+ \log {\epsilon}}2,\] 
    respectively, where Alice sends one qudit of dimension $M_i$ in the $i$-th round.
    When Alice additionally sends $c$ classical bits, the bounds become
    \[2m + c \ge  {\log |\cM|+ \log {\epsilon}},~~~\text{ and }~~2\log \left(\prod_{i=1}^k M_i\right)  + c \ge { \log |\cM|+ \log {\epsilon}}.\] 
\end{corollary}

We give some remarks. The above lemmas consider the average-case probability for input $x$, while the previous (both classical and quantum) versions~\cite{GT00,DTT10,NS06} consider the case that the success probability is at least $\epsilon$ for any input $x$. This caused a significant loss in the resulting security in the first AI-QROM bound~\cite{HXY19}, or call for the random-self-reducibility in the preprocessing DL security~\cite{CK18}. Thanks to this average-case feature, we exclude the random-self-reducibility in the proofs.

\section{Lower Bounds in the Classical Generic Group Model}\label{sec:classical}
\subsection{Generic Group Model}\label{subsec:GGM}
We first define the generic group model (GGM) of Maurer~\cite{Maurer05}, also known as the type-safe model~\cite{Zhandry22a}. 
Let $N$ be the known prime order\footnote{We can extend to the composite-order setting easily.} of our interested finite cyclic group $\cG\cong \Z_N$ with a generator $g$. A generic algorithm $\cA$ in this model is given by a circuit with the following features:
\begin{itemize}
    \item There are two types of wires: bit wires and (group) element wires. Bit wires take a bit in $\{0,1\}$, whereas element wires take an element in $\Z_N \cup \{\bot\}$. For an element wire containing $x$, we write $g^x$ to denote this wire to distinguish it from a bit string.
    \item There are bit gates that map bits to bits, which cannot take element wires as input.
    \item There are three special gates called \emph{element gates} that can access the element wires as follows:
    \begin{description}
        \item[Labeling Gate.] It takes $\lceil \log_2 N\rceil$ bit wires and interprets them as an element in $x \in \Z_N$ as input, and outputs an element wire $g^x$. If there is no corresponding element $x \in \Z_N$ to the input wires, it outputs an element wire containing $\bot$.
        \item[Group Operation Gate.] It takes two element wires containing $g^{x},g^y$ and a single bit wire containing $b$ as input. If both $g^x,g^y$ are not $\bot$, it outputs an element wire containing $g^{x+by}$.\footnote{One may define this gate differently, e.g., $(g^x,g^y,a,b) \mapsto g^{ax+by}$, but it does not make any change to our result.} Otherwise, it outputs an element wire containing $\bot$.
        \item[Equality Gate.] It takes two element wires as input. If both wires contain the same element $g^x \neq \bot$, it outputs a bit wire containing $1$. In all other cases including $\bot$ inputs, the output is $0$.
    \end{description}
\end{itemize}
An algorithm $\cA$ in this model is called a GGM algorithm and is usually denoted by $\cA^\cG$. 
The cost metric for the algorithms, denoted by \emph{the group operation complexity}, counts the number of labeling and group operation gates used in the circuit, and all other gates are considered free. 

We assume that the element gates have some orders so they can be applied sequentially (along with required bit gates).\footnote{Given the circuit, such an order can be found using the breadth-first search.} 
We also assume that GGM algorithms never make two equality gates with the same input wires. 
This ensures that for a GGM algorithm taking $m$ element wires as input and with the group operation complexity $T$, the number of group operation gates $T$, the number of equality gates less than or equal to $\binom{m+T}{2}$.
We further assume that the description of the GGM algorithm contains the order of the element gates so that they can be applied in order (ignoring bit gates). 

\begin{remark}[Relations to the other generic group models.]
A different model for generic group algorithms is suggested by Shoup~\cite{Shoup97}. The results for known-order GGM algorithms in this paper can be extended to Shoup's generic group model. 
This is because this paper focuses on the cryptographic assumptions that can be described as a single-stage game, where the generic equivalence between two models is known~\cite{Zhandry22a}. We place the detailed theorem with the proof in~\Cref{sec:equiv} for completeness. We note that our proof can be extended to the Shoup-style GGM directly, as shown in~\Cref{subsec: RRunknown}. However, this makes the proof involved, and the main body focuses on the Maurer-style model for a simpler exposition.
\end{remark}

\subsubsection{Variations: Maintaining Polynomials}
Before proceeding to the classical lower bounds in the generic group model, we give a variation of GGM algorithms, which maintains the polynomials representing the elements and information that it achieved.
We assume that the input to the GGM algorithm is specified by polynomials $P_1,...,P_m\in \Z_N[X_1,...,X_t]$ for some formal variables $X_1,...,X_t$ corresponding to the hidden values. For example, in the discrete logarithm problem, $X_1$ specifies the problem instance $g^x$, and the input is specified by $P_1=1,P_2=X_1$.
We occasionally identify a polynomial $P=a_1X_1+...+a_tX_t+b$ as a vector $(b,a_1,...,a_t) \in \Z_N^{t+1}$ (recall $N$ is prime, which makes $ \Z_N^{t+1}$ a vector space.) especially when we discuss the linear algebra notions.

Given the polynomial representations of inputs, we maintain a list $\cP$ of a pair of the element wire and polynomial called \emph{the polynomial list} and a counter $c$, and it behaves as follows.
\begin{itemize}
    \item As an initialization, set $\cP$ as an empty list. For each input element wire $w$ containing a group element corresponding $P_i$ for $i\in [m]$, store $(w,P_i)$ in the $i$-th row of $\cP$. Set $c \leftarrow m$.
    \item For a labeling gate in the circuit of $\cA$ with input representing $a \in \Z_N$ and output element wire $w$, set $c\leftarrow c+1$, and store $(w,P_c:=a)$ in the $c$-th row of $\cP$.
    \item For a group operation gate with two element wires $w_1,w_2$ and a bit wire containing $b$ as input and output wire $w$ appears, find $i,j\le c$ such that $i,j$-th rows of $\cP$ are $w_1,w_2$. Set $c\leftarrow c+1$, compute $P_c:=P_i +(-1)^b P_j$, and store $(w,P_c)$ in the $c$-th row of $\cP$.
\end{itemize}

The equality gates are dealt with differently, by maintaining \emph{the zero sets} $\cZ$ that is initialized as an empty set. For an equality gate $\eq$ with two input element wires $w_1,w_2$ and output $1$ (i.e., they are equal), we find $i,j$-th rows of $\cP$ containing $w_1,w_2$ and call $g$ by \emph{collision}; since no two equality gates have the same inputs, we also call $(i,j)$ as a collision ambiguously. We process each collision as follows. We do nothing for the equality gates outputting $0$.
\begin{itemize}
    \item If $P_i=P_j$ as a polynomial over $\Z_N$, then the collision is called \emph{trivial}, and do nothing.
    \item If an equality query finds a nontrivial collision $(i,j)$, then write $|\cZ|=z$ and $\cZ=\{Q_i\}_{i\in[z]}$, check if there exists $\veca= (a_1,...,a_z) \in \Z_N^z$ such that 
    \begin{align}\label{eqn:GGM_equalgate_zeropoly}
    P_i-P_j = a_1 Q_1 + ... + a_z Q_z
    \end{align}
    holds as a polynomial. If there is no such $\veca$, updates $\cZ\leftarrow \cZ \cup \{P_i - P_j\}$. 
    We call the collision $(i,j)$ \emph{informative}, and otherwise \emph{predictable.}
\end{itemize}
Note that the notion of informative collisions is similar to the \emph{useful} queries in~\cite{Yun15} in the search-by-hyperplane problem, but our definition is purely linear-algebraic and direct. It just says that the new informative collision must not be included in the span of the previous collisions.

The informative collisions are sufficient for describing the behavior of the GGM algorithm, as shown in the following lemma, proved in~\Cref{missing_GGM}.
\begin{lemma}\label{lem:simulation_without_queries}
    Let $\cA$ be a GGM algorithm. Given a description of the circuit for $\cA$ and the zero set $\cZ$ for the given input, the polynomial list of $\cA$ right before its termination can be computed without using the element gates, i.e., computed by a Boolean circuit.
\end{lemma}

The following auxiliary lemma is a generalization of the Schwartz-Zippel lemma, which could be of independent interest. It gives another alternative proof for the MDL lower bound presented in~\Cref{app:MDL} with its proof.
\begin{lemma}\label{lem: SZinformative}
    Suppose the hidden variables $x_1,...,x_t$ are uniform in $\Z_N$, and the group elements during the execution of the algorithm always correspond to the linear polynomials in $\Z_N[X_1,...,X_t].$
    For any equality gate for $w_i,w_j$, the probability that it induces an informative collision is at most $1/N$.
\end{lemma}

\subsection{The Discrete Logarithm Problem and Friends}
We first prove the following well-known generic lower bound for the DL problem.

\begin{problem}
    A discrete logarithm (DL) problem for a cyclic group $\cG$ of order $p$ with a generator $g$ asks to find $x$ given $(g,g^x)$ for uniformly random $x \in \{0,...,p-1\}$. An $m$-multiple DL ($m$-MDL) problem asks to find $\vecx=(x_1,...,x_m)$ given $(g,g^{x_1},...,g^{x_m})$ for uniformly random $\vecx \in \{0,...,p-1\}^m$. In the (Q)GGM, the group is fixed a priori, and the inputs are stored in the element registers.
\end{problem}
\begin{theorem}\label{thm: GGM_DL}
    Let $\cG$ be a cyclic group of prime order.
    Let $\cA_\dl$ be a DL algorithm in the GGM having at most $T$ group operation gates, then the following holds:
    \[
        \Pr_{\cA_\dl,x}\left[
            \cA_\dl^{\cG}(g,g^x) \rightarrow x
        \right]
        =
        O\left(
            \frac{T^2}{|\cG|}
        \right).
    \]
\end{theorem}
\begin{proof}
    Let $p=|\cG|$ and $\epsilon$ be the success probability of $\cA_\dl$. We make the following modifications: 
    For $z \gets \cA_\dl^\cG(g,g^x),$ we let the algorithm make the labeling gate on input $z$ and apply the equality gate on input $(g^z,g^x)$ to find a collision at the end, so that $\cA$ always finds an informative collision with probability at least $\epsilon$. Including this procedure, we assume that the algorithm makes $C=T+1$ group operations.
    The algorithm $\cA_\dl$ may be randomized by taking a random string $r$ as a seed.

    Now, we construct a pair of encoding and decoding protocols for $\cM=[p]$ and a set $R$ of seed $r$. For $x \in [p]$, the protocols are defined as follows.
    \begin{description}
        \item[$\Encode(x,r)$:] It runs $\cA_\dl^{\cG}(g,g^x)$ with randomness $r$ and outputs the equality gate $c$ with input $(i,j)$ that is the lexicographically first informative collision, i.e., for any other informative collision $(i',j')$, it holds that $i<i'$, or $i=i'$ and $j<j'$. If there is no informative collision, it outputs a special symbol $c=\bot$.
        \item[$\Decode(c,r)$:] If $c=\bot$, it outputs a random value in $[p]$. 
        Otherwise, it constructs a sub-circuit $\cA'_\dl$ of $\cA_\dl$ by cutting out the gates after the equality gate $c$ corresponding to the first informative collision $(i,j)$. 
        We associate the group element $g^{ax+b}$ with a polynomial $aX+b \in \Z_p[X]$.
        By~\Cref{lem:simulation_without_queries}, the corresponding polynomials $P_i=a_iX+b_i$ and $P_j=a_jX+b_j$ can be computed without using the element wires. Then it returns $z=-(b_i-b_j)/(a_i-a_j) \bmod p$ as an output.
    \end{description}

    We prove this protocol is correct with a probability of at least $\epsilon$, or whenever $\cA_\dl$ finds $x$. 
    In this case, the encoder finds a collision $c$ with input $(i,j)$, which is informative only when 
    \[
        (a_ix+b_i=a_jx+b_j \bmod p)\land \left((a_i,b_i)\neq (a_j,b_j)\right) ~~\Longleftrightarrow ~~x = -\frac{b_i-b_j}{a_i-a_j} \bmod p.
    \]
    Thus, given $c \neq \bot$, the decoder always finds the correct answer $x$, i.e., the protocol succeeds with probability at least $\epsilon$. 

    Now we compute the encoding length of the protocol.
    Since $\cA_\dl$ obtains at most $C+2$ group elements including inputs, the encoding space $\cC$ has the cardinality $\binom{C+2}{2} +1 \le (T+3)^2/2$.
    By~\Cref{lem: compression}, we have the following inequality
    \[
    \log \epsilon + \log |\cG| \le \log |\cC| \le  \log \left(\frac{(T+3)^2}{2}\right)
    ~~\Longrightarrow~~
    \epsilon =
    O\left ( 
        \frac{T^2}{|\cG|}
    \right)
    \]
    which concludes the proof.\ifnum\llncs=1 \qed \fi
\end{proof}

It can easily be extended to the multiple-instance DL problem with small adjustments.
For a positive integer $m$, we write $g^{\vecx}$ to denote $(g^{x_1},...,g^{x_m})$.

\begin{theorem}\label{thm:MDL_GGM}
    Let $\cG$ be a cyclic group of prime order.
    Let $\cA_{m\text{-}\mdl}$ be an $m$-MDL algorithm in the GGM having at most $T$ group operation gates. It holds that:
    \[
        \Pr_{\cA_{m\text{-}\mdl},\vecx}\left[
            \cA_{m\text{-}\mdl}^{\cG}(g,g^{\vecx}) \rightarrow \vecx
        \right]
        =
        O\left( \left( \frac{e(T+2m+1)^2}{2m|\cG|}\right)^m \right).
    \]
\end{theorem}
\begin{proof}
    Let $p=|\cG|$ and $\epsilon$ be the success probability of $\cA_{m\text{-}\mdl}$.
    With a similar modification, we assume that $\cA_{m\text{-}\mdl}$ finds at least $m$ informative collisions with probability at least $\epsilon$ using $C=T+m$ group operation complexity.
    We associate a group element $g^{a_1 x_1 + ... + a_m x_m +b}$ with a polynomial $a_1 X_1 + ... + a_m X_m + b \in \Z_p [X_1,...,X_m].$
    We additionally need the following result from linear algebra.
    \begin{fact}\label{fact:MDL}
        Given $m$ informative collisions, there is a polynomial time algorithm to find the unique assignments $(X_1,...,X_m)=(x_1,...,x_m)$ making the given collisions informative.
    \end{fact}
    The proof can be found in~\Cref{missing_GGM}.
    
    $\cA_{m\text{-}\mdl}$ may be randomized using a random seed $r$.
    We construct encoding and decoding protocols for $\cM=[p]^m$ and a set $R$ of the seed $r$. For input $\vecx \in \Z_p^m$, the protocols are defined as follows.
    \begin{description}
        \item[$\Encode(\vecx,r)$:] 
        It runs $\cA_{m\text{-}\mdl}^\cG(g,g^{\vecx})$ with randomness $r$ and collects the lexicographically first $m$ informative collision gates $c_k$ with input $(i_k,j_k)$ for $k \in [m]$.
        If $m$ informative collisions are found during the execution, 
            it outputs $\vecc=(c_1,...,c_m).$
        Otherwise, 
            it outputs a symbol $\bot$.
        \item[$\Decode(\vecc,r)$:]
        If $\vecc=\bot$, it outputs a random value in $[p]^m$. 
        Otherwise, it parses $\vecc=(c_1,...,c_m)$ and constructs a sub-circuit $\cA'_{m\text{-}\mdl}$ of $\cA_{m\text{-}\mdl}$ by cutting out the gates after the $m$-th informative collision gate (corresponding to $c_m$).
        By~\Cref{lem:simulation_without_queries}, it recovers the polynomial list. 
        Using $\vecc$ and~\Cref{fact:MDL}, it finds and outputs the assignment $\vecz=(z_1,...,z_m)$ of $(X_1,...,X_m)$.
    \end{description}

    It is obvious that if $\cA_{m\text{-}\mdl}$ finds $\vecx=(x_1,...,x_m)$ then the decoder correctly recovers $\vecx$, which happens with probability at least $\epsilon$. We focus on the encoding size below. 
    Let $B=\binom{C+m+1}{2}$ be the upper bound of the number of equality queries.
    The bit-length for describing the $m$ informative collision (or $\bot$) is less than $\log \left(\binom{B}{m} + 1 \right) $, which is bounded by
    \begin{align*}
        \log \binom{B+1}{m} 
        \le
        m \log\left( \frac{e(B+1)}{m} \right) 
        \le 
        m \log\left( \frac{e(T+2m+1)^2}{2m} \right),
    \end{align*}
    where we use $B+1 \le \frac{(C+m+1)^2}{2} \le \frac{(T+2m+1)^2}{2}$.
    By~\Cref{lem: compression}, we have
    \[
        \log \epsilon + m\log |\cG| \le \log |\cC| \le  
        m \log\left( \frac{e(T+2m+1)^2}{2m} \right)
    \]
    which can be rewritten as follows
    \[
        \epsilon =
        O\left ( 
        \left(\frac{e(T+2m+1)^2}{2m|\cG|}\right)^m
        \right),
    \]
    as we desired.\ifnum\llncs=1 \qed \fi
\end{proof}

\subsection{Oracle Problems in the GGM}
This section extends the lower bounds relative to the oracle. We first consider the following problems.
\begin{problem}
    In the gap DL (gap-DL) problem, the adversary is given $(g,g^x)$ as input and is asked to find $x$, having access to the decisional Diffie-Hellman (DDH) oracle: $O_{\ddh}:(g^x,g^y,g^z) \mapsto \delta_{xy,z}$. In the gap computational Diffie-Hellman (gap-CDH) problem, the adversary is given $(g,g^x,g^y)$ and is asked to output $g^{xy}$ with the DDH oracle access.
\end{problem}

\begin{theorem}\label{thm: GGM_GapDL}
    Let $\cG$ be a cyclic group.
    Let $\cA_\gapdl$ be a gap-DL algorithm in the GGM having at most $T$ group operation gates and making $T_{\ddh}$ queries to the DDH oracle, then the following holds:
    \[
        \Pr_{\cA_\gapdl,x}\left[
            \cA_\gapdl^{\cG,O_\ddh}(g,g^x) \rightarrow x
        \right]
        =
        O\left(
            \frac{T^2+T_\ddh}{|\cG|}
        \right).
    \]
\end{theorem}
\begin{proof}[{\ifnum\llncs=0 Proof sketch\else Sketch\fi}]
    We extend the notion of collisions to include the DDH oracle answers that output $1$. 
    By a similar modification as the previous section, we can assume that the algorithm finds at least one informative collision. If it is an answer from the DDH oracle, then it specifies the equation $(aX+b)(cX+d)=eX+f$ for some $a,b,c,d,e,f$. It has at most two solutions; thus, the encoding includes one more bit to specify the correct solution. The number of collisions is bounded by $\binom{T+3}{2}+T_{\ddh}$.    
    The other parts of the proof are identical.
    \ifnum\llncs=1 \qed \fi
\end{proof}

\begin{theorem}\label{thm: GGM_GapCDH}
    Let $\cG$ be a cyclic group.
    Let $\cA_\gapcdh$ be a gap-CDH algorithm in the GGM having at most $T$ group operation gates and making $T_{\ddh}$ queries to the DDH oracle, then the following holds:
    \[
        \Pr_{\cA_\gapcdh,x}\left[
            \cA_\gapcdh^{\cG,O_\ddh}(g,g^x,g^y) \rightarrow g^{xy}
        \right]
        =
        O\left(
            \frac{T^2+T_\ddh}{|\cG|}
        \right).
    \]
\end{theorem}
The proof is almost identical and placed in~\Cref{missing_GGM}.

\begin{problem}
    In the one-more-DL ($\omdl$) problem, the adversary is given access to the challenge oracle $O_\chal$ that outputs $g^{x_i}$ for an unknown $x_i$ and to the DL oracle $O_\dl:g^x \mapsto x$. 
    The number of DL oracle queries $q$ must be less than the number of challenge queries $t$. The adversary aims to find all answers to the challenges. More generally, in the $n$-out-of-$m$-more-DL ($(n,m)\mmdl$) problem, it must hold that $t=q+m$, and the adversary needs to find $q+n$ solutions to the challenges among $q+m$ challenges.
\end{problem}

\begin{theorem}\label{thm: GGM_OMDL}
    Let $\cG$ be a cyclic group.
    Let $\cA_{(m,n)\mmdl}$ be an $n$-out-of-$m$-more-DL algorithm in the GGM having at most $T$ group operation gates and making $q$ queries to the DL oracle, then the following holds:
    \[
        \Pr_{\cA_{(n,m)\mmdl},x}\left[
            \cA_{(n,m)\mmdl}^{\cG,O_\chal,O_\dl}(g) \text{ solves }(n,m)\mmdl
        \right]
        =
        O\left(
            \left(
                \frac{e(T+m+n+1)^2}{|\cG|}
            \right)^n
        \right).
    \]
    In particular, the advantage against $\omdl$ is $O\left(\frac{T^2}{|\cG|}\right).$
\end{theorem}
\begin{proof}
    Suppose that the $t=m+q$ challenges are $g^{x_1},...,g^{x_t}$. We construct an encoding for $\vecx=(x_1,...,x_t)$. We assume that the algorithm is deterministic.
    Regarding informative collisions, we include the DL oracle answers as the collision. If the DL oracle outputs $z$ for input $g^{P}$, we regard $P-z$ as a collision. Since the algorithm finds $n+q$ solutions, it must find $n+q$ informative collisions (including the DL oracle outputs).
    We assume that the algorithm never queries to the DL oracle that the answer induces a trivial collision. This means there are $n$ informative collisions that are not from the DL oracle queries.

    We need the following simple fact from linear algebra.
    \begin{fact}\label{fact: indep}
        Given $a$ linear independent linear equations over $b$ variables for $a<b$. There are $b-a$ variables such that the linear equations are still independent after fixing them to some values.
    \end{fact}
    
    The procedures are as follows.
    \begin{description}
        \item[$\Encode(\vecx)$:] 
        It runs $\cA_{m\mmdl}^{\cG,O_\chal,O_\dl}(g)$ and collects the lexicographically first $t$ informative collisions, which could be the equality gate or the DL oracle answer.
        If $n+q$ informative collisions are found during the execution, 
        it outputs $\vecc=(c_1,...,c_n)$ that denote the informative equality gates and the DL oracle answers $\vecz=(z_1,...,z_q).$
        By~\Cref{fact: indep}, there are $m-n$ $x_i$'s such that revealing them does not hurt the linear independence of the informative collisions. 
        Finally, those $x_i$'s, denoted by $\vecw$ become a part of the encoding.
        Otherwise, 
            it outputs a symbol $\bot$.
        \item[$\Decode(\vecc,\vecz,\vecw)$:]
        It runs $\cA_{m\mmdl}^{\cG,O_\chal,O_\dl}(g)$ to recover $n+q$ informative collisions.
        Given these equations, the decoder can recognize the indices for $\vecw$. It recovers $\vecw,$ and plugs them in the informative collisions. The informative collisions become $n+q$ linear equations over $n+q$ variables, so that it can recover $\vecx$.
    \end{description}
    The length of the encoding is bounded by
    \[
    \log \binom{\binom{T+m+n}{2}}{n} + q\log |\cG|+ (m-n)\log |\cG|+O(1)
    \]
    which must be larger than $(m+q)\log |\cG| +\log \epsilon$ for the success probability $\epsilon$ by~\Cref{lem: compression}. This gives
    \[
        n\log \left(
            \frac{e(T+2m+1)^2}{2n}
        \right) \ge n\log |\cG| + \log \epsilon
    \]
    which implies
    \[
    \epsilon = O\left(
            \left(
                \frac{e(T+m+n+1)^2}{|\cG|}
            \right)^n
        \right),
    \]
    concluding the proof.
    \ifnum\llncs=1 \qed \fi
\end{proof}

\begin{remark}
    Extending the results to high-degree variants like $m$-CDH problems is not trivial. We believe with some algebraic geometry reasoning like Bézout theorem, as in~\cite{AGK20}, the high-degree variants can be proven in essentially the same way.
\end{remark}

\section{Lower Bounds in the Unknown-order GGM}\label{sec: unknownGGM}
We extend the generic group to the unknown-order setting. As the order is unknown, we should consider the distribution of the order. 

Let $\cG$ be a cyclic group of order $N$, where the distribution of $N$ will be specified later.
We assume that $N$ is unknown to the algorithm except for its bit length. The other interface of the generic algorithms is identical to the (known-order) GGM. In particular, the assumption that the group operation only allows to compute $(g^x,g^y) \mapsto g^{x+y}$ is important.\footnote{This was also used in the related works~\cite{DK02,Sut07} in the unknown-order GGM.}
Note that the algorithm cannot extract any information from the element wire containing $\bot$; for example, the equality gate involving $\bot$ always outputs $0$.

\begin{remark}\label{rem:unknown_equiv}
    We found that the known equivalence proof between the generic group models does not extend to the unknown-order group setting. Therefore, we include the random-representation GGM proof at~\Cref{subsec: RRunknown}.
\end{remark}

\paragraph{Polynomial Representations}
As in the known-order GGM, we give the polynomial representations for each group element. However, as the group order is unknown, we choose the polynomials from $\Z[X_1,...,X_t]$ \emph{without modulus} for the formal variables $X_1,...,X_t$ corresponding to the hidden values. In particular, if the algorithm takes no input, the representations could be in just $\Z$ without formal variables.

We extend the notion of informative collisions appropriately. 
We maintain the zero set $\cZ$ and process each collision $(i,j)$ with inputs corresponding to polynomials $P_i,P_j$ (i.e., the equality gate outputting 1) as follows:
\begin{itemize}
    \item If $P_i=P_j$ as a polynomial, then the collision is called \emph{trivial}, and do nothing.
    \item If an equality query finds a nontrivial collision $(i,j)$, then 
    check if $P_i-P_j$ is included in $\Z$-span of $\cZ$; recall the we only checked $\Z_N$-span in the known-order case.
    If it is not true, update $\cZ\leftarrow \cZ \cup \{P_i - P_j\}$. 
    We call the collision $(i,j)$ \emph{informative}, and otherwise \emph{predictable.}
\end{itemize}

\subsection{Order-finding in the Unknown-order GGM}\label{subsec:uGGM}

We consider the following problem.
\begin{problem}
    Let $\cD_{\pprime}^{(n)}$ be a uniform distribution over the set of $n$-bit primes. An order-finding problem over $\cD_\pprime^{(n)}$ in the GGM is defined as follows. First, a random $N$ is sampled from $\cD_\pprime^{(n)}$. The adversary in the unknown-order GGM for the group $\cG_N$ of order $N$ is asked to output $N$. A product-order-finding problem over $\cD_\pprime^{(n)}$ is similarly defined, but the two distinct primes $p,q$ are sampled and $N:=pq.$
\end{problem}

This problem is studied in~\cite{Sut07} in detail. In particular, the generic order-finding algorithm with the $O(\sqrt{N/\log\log N})$ group operation complexity suggested in~\cite[Section 4]{Sut07}, and the lower bound of $\Omega(N^{1/3})$ (for the prime-order case) is proven in the same thesis. We reprove this bound using our method.

\begin{theorem}\label{thm: uGGM-order}
    Let $\cA_\ord$ be an order-finding algorithm over $\cD_\pprime^{(n)}$ in the GGM with the group operation complexity $T$. It holds that
    \[
        \Pr_{\cA_\ord,N} \left[
            \cA_\ord^{\cG_N}(g) \rightarrow N
        \right]
        =O\left(
            \frac{T^3}{2^{n}}
        \right).
    \]
    In particular, any generic order-finding algorithm with a constant success probability must make $\Omega\left(N^{1/3}\right)$ group operations.
\end{theorem}
\begin{proof}
    Let $\epsilon$ be the success probability of $\cA_\ord$. 
    For simplicity, we assume that $\cA_\ord$ is deterministic. 
    
    We consider the integer representations corresponding to the elements of $\cA_\ord$ because of the unknown order and no indeterminate value.
    In this case, an informative collision $(i,j)$ must specify two integers $x_i,x_j$ such that $x_i-x_j$ is a multiple of the order $N$.
    The integer $x$ appearing in this list must satisfy
    \[
        |x| \le 2^T N
    \]
    because a group operation only increases the number twice.
    This gives that the number of $n$-bit primes divisors of $x-y$ for some $x,y$ appearing in the list is bounded by
    \[
        \log_{2^n} (2^TN) = O\left(\frac{T}{\log N}\right).
    \]
    
    We consider the following modification: If $\cA_\ord$ outputs $z$, we let the algorithm make the labeling gate on inputs $z,0$ and apply the equality gate on input $(g^0,g^z)$ to find a collision at the end with probability at least $\epsilon.$

    Now, we construct the following encoding-decoding pair.
    \begin{description}
        \item[$\Encode(N)$:] It runs $\cA_\ord^{\cG_N}(g)$ and computes the equality gate $c$ with input $(i,j)$ that is the lexicographically first informative collision. Let $x_i,x_j$ be the corresponding integer representations. It factorizes $x_i-x_j$ and lets $p_1,...,p_K$ be the $n$-bit prime divisors. It outputs $(c,\ell)$ where $p_\ell = N$ if exists. Otherwise, it outputs a special symbol $c=\bot$.
        \item[$\Decode(c,\ell)$:] If $c=\bot$, it outputs a random sample from $\cD_\pprime^{(n)}$. Otherwise, it recovers $x_i,x_j$, computes and outputs the $\ell$-th prime factor $N'$.
    \end{description}
    The correctness is analogous. The size of the encoding is $\log \binom{T}{2}+\log (K) + O(1)$, and we have $K=O(T/\log N)$. This gives
    \[
        \log \binom{T}{2}+\log (K) + O(1) \ge \log \left(\frac{2^n}{n}\right)+O(1) + \log \epsilon ~\Longrightarrow~ \epsilon = O \left(\frac{T^3}{N}\right)
    \]
    applying~\Cref{lem: compression} and $2^{n-1}\le N\le 2^n$.
    \ifnum\llncs=1 \qed \fi
\end{proof}

We can prove the analogous result for the product of two primes. The proof is essentially identical, except that we need to encode two prime factors using two informative collisions.
\begin{theorem}\label{thm: uGGM-RSA}
    Let $\cA_\ord$ be a product-order-finding algorithm over $\cD_\pprime^{(n)}$ in the GGM with the group operation complexity $T$. It holds that
    \[
        \Pr_{\cA_\ord,p,q} \left[
            \cA_\ord^{\cG_{pq}}(g) \rightarrow pq
        \right]
        =O\left(
            \frac{T^4}{2^{2n}}
        \right).
    \]
    In particular, any generic order-finding algorithm with a constant success probability must make $\Omega\left(N^{1/4}\right)$ group operations for $N=pq.$
\end{theorem}

\subsection{Root Extraction and Repeated Squaring Problems}
We prove similar lower bounds for the (strong) root extraction and the repeated squaring in the unknown-order GGM.

\begin{theorem}\label{thm: uGGM_rootext}
    Let $\cA$ be an algorithm in the GGM with the group operation complexity $T$. 
    Suppose that $N$ is sampled from $\cD_\pprime^{(n)}$. 
    It holds that
    \[
        \Pr_{\cA,N,x} \left[
            g^{ey}=g^x :\cA^{\cG_{N}}(g,g^x) \rightarrow (e,g^y)
        \right]
        =O\left(
            \frac{T^3}{2^{n}}
        \right).
    \]
\end{theorem}
\begin{proof}[{\ifnum\llncs=0 Proof sketch\else Sketch\fi}]
    If there is an informative collision when running $\cA$, we can apply the same encoding as in the previous section. We show that if the algorithm finds the root $g^y,$ then it finds an informative collision. 
    
    Suppose that there is no informative collision during the execution for given input $(g,g^x)$. Then, the polynomial corresponding to $g^y$ must be $Y=aX+b \in \Z[X].$ 
    The correctness implies that $eax+eb=x \bmod N.$ 
    To do so, either $N|ea-1, N|eb$ or $x=eb/(ea-1) \bmod N$ must hold.
    The first case implies that $N|b$ (otherwise $ea-1$ is not divided by $N$), and
    since $|b|| \le 2^{T} N,$ this event only happens with probability at most $O(Tn/2^n).$
    The second case holds with probability $1/N$. 
    In other words, except this probability, the algorithm finds an informative collision.
    \ifnum\llncs=1 \qed \fi
\end{proof}

\begin{theorem}\label{thm: uGGM_repeatsq}
    Let $\cA$ be an algorithm in the GGM with the group operation complexity $T$. 
    Suppose that $N$ is sampled from $\cD_\pprime^{(n)}$. 
    Let $t >T$ be a positive integer.
    It holds that
    \[
        \Pr_{\cA,N,x} \left[
            \cA^{\cG_{N}}(g)\rightarrow g^{2^t}
        \right]
        =O\left(
            \frac{(T+t)^3}{2^{n}}
        \right).
    \]
\end{theorem}
\begin{proof}[{\ifnum\llncs=0 Proof sketch\else Sketch\fi}]
    If the algorithm $\cA$ outputs $h$, we can compute $g^{2^t}$ using $t$ group operations and check if $h=g^{2^t}$. Also, the integer representation corresponding to $h$ must be smaller than $2^T$, thus it should be the informative collision. Using this, we can construct an encoding algorithm for $N$ as in the previous section, proving the desired result.
    \ifnum\llncs=1 \qed \fi
\end{proof}

\section{Lower Bounds in the Quantum GGM}
\subsection{Quantum Generic Group Models}

\subsubsection{Basic Quantum Generic Group Model}
We define the quantum generic group model (QGGM) extending the model in~\Cref{subsec:GGM}, following the formalization in~\cite{HYY23}. Let $\cG$ be a cyclic group of order $N$ with a generator $g$. A quantum generic group algorithm $\cA$ works similarly to a generic group algorithm but is defined on the registers holding qubits or superpositions of elements.

We first consider a rudimentary model, denoted by \emph{the basic QGGM}, where group operations only work on two a priori fixed registers. As we look for the logarithmic lower bounds, we do not allow the quantum labeling gate and give a quantum inversion gate as a unit.
An algorithm in the basic QGGM is defined as follows.
\begin{itemize}
    \item There are two registers: qubit and element registers holding superpositions of some information. Qubit registers take a set of bits $\{0,1\}$ as the computational basis. In contrast, element registers take a set of elements $x\in \cG \cup \{\bot\}$ as the computational basis, which is denoted by $g^x$; sometimes $\bot$ is also written in this form though there is no corresponding $x$. The algorithm arbitrarily appends a new element register initialized by $\ket{g}$.
    \item There are (arbitrary) quantum gates that map qubits to qubits, which cannot take element registers as input.
    \item There are two special gates called \emph{element gates} that can access the element wires as follows:
    \begin{description}
        \item[Group Operation Gate.] It takes two element registers $\vecX,\vecY$ and a single qubit register $\vecB$ and applies the unitary $U_{\Gop}$ that works on the computational basis as follows:
        \begin{align}\label{eq: QGGM gop}
            U_{\Gop}:\begin{cases}
                \ket{b}_{\vecB} \ket{g^x,g^y}_{\vecX,\vecY} 
                \mapsto\ket{b}_{\vecB} \ket{g^{x+by},g^{ y}}_{\vecX,\vecY} 
                &\text{ if }g^x,g^y \neq \bot,\\
                \ket{b}_{\vecB} \ket{g^x,g^y}_{\vecX,\vecY} 
                \mapsto\ket{b}_{\vecB} \ket{g^x,g^y}_{\vecX,\vecY} 
                &\text{ otherwise.}
            \end{cases}
        \end{align}
        \item[Inverse-Operation Gate.] It takes two element registers $\vecX,\vecY$ and a single qubit register $\vecB$ and applies the unitary $U_{\Ginv}$ that works on the computational basis as follows:
        \begin{align*}\label{eq: QGGM ginv}
            U_{\Ginv}:\begin{cases}
                \ket{b}_{\vecB} \ket{g^x,g^y}_{\vecX,\vecY} 
                \mapsto\ket{b}_{\vecB} \ket{g^{x-by},g^{ y}}_{\vecX,\vecY} 
                &\text{ if }g^x,g^y \neq \bot,\\
                \ket{b}_{\vecB} \ket{g^x,g^y}_{\vecX,\vecY} 
                \mapsto\ket{b}_{\vecB} \ket{g^x,g^y}_{\vecX,\vecY} 
                &\text{ otherwise.}
            \end{cases}
        \end{align*}
        \item[Equality Gate.] It takes two element registers $\vecX,\vecY$ and a single qubit register $\vecB$. It then applies the unitary operation $U_{\Geq}$ that works on the computational basis as follows:
        \begin{align*}
            U_{\Geq}:\begin{cases}
                \ket{b}_{\vecB} \ket{g^x,g^y}_{\vecX,\vecY} 
                \mapsto\ket{b\oplus \delta_{x,y}}_{\vecB} \ket{g^x,g^y}_{\vecX,\vecY} 
                &\text{ if }g^x,g^y \neq \bot,\\
                \ket{b}_{\vecB} \ket{g^x,g^y}_{\vecX,\vecY} 
                \mapsto\ket{b}_{\vecB} \ket{g^x,g^y}_{\vecX,\vecY} 
                &\text{ otherwise,}
            \end{cases}
        \end{align*}
        where $\delta_{x,y}=1$ if $x=y$ and $0$ otherwise.
    \end{description}
    \item We allow the intermediate measurements for registers. When we apply the measurements on all of $\vecB,\vecX,\vecY$ right before applying the element gates, we call them \emph{classical}. An element gate that is not classical is called quantum. For simplicity, we allow the classical labeling gate that does not much affect the result.
    \begin{description}
        \item[Classical Labeling Gate.] It takes $\lceil \log_2 N\rceil$ qubit registers, measures it, and interprets them as an element in $x \in \Z_N$. It appends a new element register holding $\ket{g^x}$. If there is no corresponding element $x \in \Z_N$ to the input wires, it outputs an element wire containing $\ket{\bot}$.
    \end{description}
\end{itemize}
A QGGM algorithm denotes an algorithm $\cA$ in this model and is occasionally written by $\cA^{\ket{\cG}}$. 
As in the classical GGM, we assume the element gates have some order to be applied sequentially with the relevant qubit gates.

The formal complexity measure of the generic algorithms is described in the next subsection.
Roughly, we count the number of \emph{quantum} element gates \emph{including the equality gates} as the cost metric. The main reason is that while the equality check between two classical data is essentially free, e.g., using hash tables, the equality check between two element registers that store superpositions is not freely done. We also note that Shor's algorithm does not use any equality query.

\subsubsection{QGGM with Coherent Indices}
Now, we consider more general operations that can coherently access the indices of registers. Let $t,w$ be positive integers. We define the $(t,w)$-QGGM similarly to the basic QGGM, but it also has \emph{qudit} registers of dimension $t$ and $w$, and the element gates are defined as follows.
\begin{itemize}
    \item There are two special gates called \emph{element gates} that can access the element wires as follows. The unspecified registers are unchanged by the operations.
    \begin{description}
        \item[Group Operation Gate.] It takes three registers $\vecB,\vecT,\vecW$ and $t+w$ element registers $\vecX_1,...,\vecX_t,\vecY_1,...,\vecY_w$ and applies the unitary $U_{\Gop}^{(t,w)}$ that works on the computational basis as follows:
        \begin{align*}
            U_{\Gop}^{(t,w)}:
                \ket{b,i,j}_{\vecB\vecT\vecW} \ket{g^{x_i},g^{y_j}}_{\vecX_i,\vecY_j} 
                \mapsto\ket{b,i,j}_{\vecB\vecT\vecW} \ket{g^{x_i+by_j},g^{y_j}}_{\vecX_i,\vecY_j}
        \end{align*}
        for $i\in[t],j\in[w],g^{x_i},g^{y_j} \neq \bot$ and $\vecX_i\neq\vecY_j$, otherwise do nothing.
        \item[Inverse-Operation Gate.] It takes three registers $\vecB,\vecT,\vecW$ and $t+w$ registers $\vecX_1,...,\vecX_t,\vecY_1,...,\vecY_w$ and applies the unitary $U_{\Ginv}^{(t,w)}$ that works on the computational basis as follows:
        \begin{align*}
            U_{\Ginv}^{(t,w)}:
                \ket{b,i,j}_{\vecB\vecT\vecW} \ket{g^{x_i},g^{y_j}}_{\vecX_i,\vecY_j} 
                \mapsto\ket{b,i,j}_{\vecB\vecT\vecW} \ket{g^{x_i-by_j},g^{y_j}}_{\vecX_i,\vecY_j}
        \end{align*}
        for $i\in[t],j\in[w],g^{x_i},g^{y_j} \neq \bot$ and $\vecX_i\neq\vecY_j$, otherwise do nothing.
        \item[Equality Gate.] It takes three registers $\vecB,\vecT,\vecW$ and $t+w$ element registers $\vecX_1,...,\vecX_t,\vecY_1,...,\vecY_w$ and applies the unitary $U_{\Geq}^{(t,w)}$ that works on the computational basis as follows:
        \begin{align*}
            U_{\Geq}^{(t,w)}:
                \ket{b,i,j}_{\vecB\vecT\vecW} \ket{g^{x_i},g^{y_j}}_{\vecX_i,\vecY_j} 
                \mapsto\ket{b\oplus\delta_{x_i y_j},i,j}_{\vecB\vecT\vecW} \ket{g^{x_i},g^{y_j}}_{\vecX_i,\vecY_j}
        \end{align*}
        for $i\in[t],j\in[w],g^{x_i},g^{y_j} \neq \bot$ and $\vecX_i\neq\vecY_j$, and do nothing other cases,
        where $\delta_{x,y}=1$ if $x=y$ and $0$ otherwise.
    \end{description}
\end{itemize}
Note that the $(1,1)$-QGGM is identical to the basic QGGM. We remark that allowing coherent access to indices is relevant to practice. Coherent access to indices means the corresponding unitary operation should be large and implemented differently from the above gates. 
Furthermore, setting $t>1$ implies that the quantum storage should store $t$ group elements, requiring a large quantum memory.
Allowing $w>1$ was studied in~\cite{Gid19}, and the estimation in~\cite{GE21} mainly used $t=1$ and $w=5$.

When we say \emph{the QGGM}, the choice of $(t,w)$ is unimportant in that context.
\begin{remark}
    We did not explicitly state that the element registers are different. This potentially allows the group operations between the registers $\vecX_i,\vecX_j$.
\end{remark}

\subsubsection{Quantum Generic Group Algorithm with Classical Preprocessing}
This paper considers the generic algorithms for the discrete logarithm that may perform classical generic computation before running the quantum parts. Formally, a generic $(C,Q)$-algorithm $\cA$ in the QGGM decomposes into two generic algorithms $\cA_c,\cA_q$ as follows.
\begin{enumerate}
    \item Given the problem instance, $\cA_c$ consists of at most $C$ classical group operation gates and arbitrarily many classical equality gates. It may have arbitrarily many qubit gates. At the end, it gives all registers to $\cA_q$.
    \item Given the registers from $\cA_c$ as input, it applies at most $Q$ quantum element gates along with arbitrarily many qubit gates. It measures the output registers and returns the measurement result as the outcome.
\end{enumerate}

The following lemma shows the classical equality gates can be safely removed. The proof can be found in~\Cref{missing_QGGM}.
\begin{lemma}\label{lem: classical_eq_remove}
    Let $\cG$ be a cyclic group of order $N$, and $p$ the \emph{smallest} prime divisor of $N$.
    For any $(C,0)$-algorithm $\cA_c$ in the (arbitrary) QGGM for $\cG$, there is another $(C,0)$-algorithm $\cA_c'$ without equality gates such that
    \[
        \Pr_{\vecx \gets [N]^m}\left[ 
            \cA_c \ket{0^n,g,g^{x_1},...,g^{x_m}} 
            =\cA_c' \ket{0^n,g,g^{x_1},...,g^{x_m}}
        \right] \ge 1 - \frac{(C+m+1)^2}{2p}.
    \]
    In particular, for generic $(C,Q)$-algorithms $\cA=(\cA_c,\cA_q)$ and $\cA'=(\cA_c',\cA_q)$ for $\cA_c'$ defined above, the outputs of two algorithms are identical with probability at least $1-\frac{(C+m+1)^2}{2p}.$
\end{lemma}

\subsection{Discrete Logarithms in the QGGM}
\paragraph{The basic QGGM.}
We begin with the DL lower bound in the basic QGGM.
\begin{theorem}\label{thm: QGGM_DL}
    Let $\cG$ be a cyclic group of order $N$ with a generator $g$. Suppose that the smallest prime divisor of $N$ is $p$. Let $\cA_\dl$ be a $(C,Q)$-algorithm in the basic QGGM, then the following holds:
    \[
        \Pr_{\cA_\dl,x\gets [N]}\left[
            \cA_\dl^{\ket{\cG}}(g,g^x) \rightarrow x
        \right]
        \le\frac{(C+2)^2}{2p} + \frac{2^{2Q}}{N}.
    \]
\end{theorem}
\begin{proof}
    Let $\epsilon$ be the success probability of $\cA_\dl$.
    Decompose $\cA_\dl=(\cA_c,\cA_q)$ as described in the previous section. 
    By~\Cref{lem: classical_eq_remove}, it suffices to consider $\cA_\dl'=(\cA_c',\cA_q)$ where $\cA_c'$ does not have any equality gates, whose output is identical to $\cA_\dl=(\cA_c,\cA_q)$ with probability $1-(C+2)^2/2p$.
    In other words, $\cA'_\dl$ solves the DL problem with a probability of at least
    \begin{align}\label{eq: prob_QDL_q}
        \epsilon' \ge \epsilon - \frac{(C+2)^2}{2p}.
    \end{align}

    Similarly to the classical case, we will construct an interactive compression protocol and apply~\Cref{lem: interactive compression}. In the protocol, Alice holds group registers and applies element gates. Bob only holds the qubit registers and \emph{delegates} all group-related operations to Alice. 
    
    We introduce the simple sub-protocols between Alice and Bob, showing that Alice sends one qubit during one element gate delegation.

    \paragraph{Subprotocol $\DGop$ for group operation gates.}
    The initial states are
    \begin{align}\label{eq: tele_init}
        \sum_{b}\beta_b \ket{b}_\vecB 
        \otimes \sum_{z,w} \alpha_{z,w}\ket{g^z,g^w}_{ \vecX\vecY}
    \end{align}
    where Alice holds the registers $\vecX,\vecY$ and Bob holds $\vecB$.

    \begin{enumerate}
        \item Bob sends his register $\vecB$ to Alice.
        Alice applies quantum group operation gates on $\vecB\vecX\vecY$ to obtain
        \begin{align}\label{eq: tele_fin}
            \sum_{b}\beta_b \ket{b}_{\vecB}
            \otimes\sum_{z,w}  \alpha_{z,w}\ket{g^{z+bw},g^w}_{\vecX\vecY}.
        \end{align}
        \item Alice returns the register $\vecB$ to Bob.
    \end{enumerate}
    In this protocol, Alice only sent one qubit and the group operation gate is applied as a result (compare~\Cref{eq: tele_init,eq: tele_fin} and \Cref{eq: QGGM gop}). 
    
    \paragraph{Subprotocol $\DGinv$ and $\DGeq$.} These are almost the same as the protocol $\DGop$. The difference is as follows.
    \begin{enumerate}
        \item Alice applies the quantum inversion-operation gate or equality gate instead of the group operation gate.
    \end{enumerate}

    Now, we return to the proof. We construct the following interactive protocol between Alice and Bob, where Alice selects $x\in [N]$ and tries to send $x$ using this protocol. Given a $(T,Q)$-algorithm $\cA'_\dl=(\cA'_c,\cA_q)$ for $T=C+2$, we consider the following protocol.

    \paragraph{Main interactive protocol.}
    Suppose that Alice chooses $x\in [N]$. 
    In the protocol, Alice and Bob try to execute the algorithm $\cA'_\dl$ together, while Alice holds all element registers and Bob holds all qubit registers. For qubit gates, Bob applies them locally without interacting with Alice.

    To apply element gates, Alice and Bob use the above protocol. For the classical preprocessing, a simpler protocol suffices.
    We give the overall protocol below. 
    \begin{enumerate}
        \item Alice prepares two element registers holding $\ket{g,g^x}$. If they are stored in the $i,j$-th element registers in $\cA'_c$'s input, Alice also stores them in the $i,j$-th element registers. \label{protocol: init}
        \item Alice and Bob together execute $\cA'_c$, with the following modifications. 
        \begin{itemize}
            \item Every qubit register is stored in Bob's memory, and every qubit gate is applied to Bob's side accordingly. Every element register is stored in Alice's memory. Alice and Bob use the same name/order of the registers as in $\cA'_c$.
            \item For each classical group operation gate that is applied to the registers $\vecB$ and $\vecX,\vecY$, Bob measures $\vecB$ in the computational basis and sends the measurement outcome $b$ to Alice. Alice applies the group operation on her registers $\vecX\vecY$ controlled on $b$, and discards $b$.
            \item Each classical labeling gate is processed analogously.
        \end{itemize}
        We make some observations on this part. 
        Alice's state is always classical during this procedure, so the measurement of Alice's registers can be ignored, and discarding bit $b$ is not problematic. 
        Alice has not sent any information to Bob until this point.
        Finally, the overall states between Alice and Bob are identical to the state after $\cA'_c(g,g^x)$, except that all qubit registers are stored in Bob's memory and all element registers are stored in Alice's memory.
        \item Alice and Bob execute $\cA_q$ together in a similar way:
        \begin{itemize}
            \item Every qubit gate is applied to Bob's registers accordingly.
            \item For each group operation gate $U_\Gop$ that is applied to the qubit register $\vecB$ and element registers $\vecX,\vecY$, Alice and Bob executes $\DGop$ on $\vecB\vecX\vecY$.
            \item Similarly, $\DGinv$ or $\DGeq$ is executed for each $U_\Ginv$ or $U_\Geq$, respectively.
        \end{itemize}
        \item Finally, Bob outputs the final output of $\cA_q$.
    \end{enumerate}
    It is not hard to see that the overall states between Alice and Bob are always identical to the corresponding intermediate states of $\cA'_\dl$ (ignoring the discarded bits). Therefore, the probability that Bob successfully recovers $x$ is exactly the same as that $\cA'_\dl$ solves the DL problem on input $(g,g^x)$.

    We then count the number of qubits sent from Alice to Bob. Alice sends a bit only when $\cA_\dl$ applies an element gate. Thus, the total number of qubits is $Q$. At this point, we can apply~\Cref{lem: interactive compression} to have the following inequality:
    \[
        \frac{\log \epsilon' + \log N}2 \le Q ~\Longrightarrow~\epsilon \le \epsilon' + \frac{(C+2)^2}{2p} \le  \frac{2^{2Q}}{N} + \frac{(C+2)^2}{2p}
    \]
    where we use~\Cref{eq: prob_QDL_q}, which completes the proof.
    \ifnum\llncs=1 \qed \fi
\end{proof}

\paragraph{The $(t,w)$-QGGM.}
We then extend the lower bounds in the QGGM for more general settings. The proof ideas are almost the same, except for the sub-protocols; Alice needs to send one qudit for appropriate dimensions.
We present the following generalization to the MDL problem.
\begin{theorem}\label{thm: QGGM_fine-grained-DL}
    Let $\cG$ be a cyclic group of order $N$ with a generator $g$. Suppose that the smallest prime divisor of $N$ is $p$. 
    Let $m$ be a positive integer.
    Let $\cA_\mdl$ be a $(C,Q)$-algorithm in the $(t,w)$-QGGM, then the following holds:
    \[
        \Pr_{\cA_\dl,\vecx\gets [N]^m}\left[
            \cA_\dl^{\ket{\cG}}(g,g^\vecx) \rightarrow \vecx
        \right]
        \le\frac{(C+m+1)^2}{p} + \frac{(2tw)^{2Q}}{N^m}.
    \]
\end{theorem}
\begin{proof}
    Applying~\Cref{lem: classical_eq_remove}, it suffices to consider the generic algorithm $\cA'_\mdl$ with no classical equality queries, which solves the MDL problem with probability at least $\epsilon'\ge \epsilon-\frac{(C+m+1)^2}{2p}$. Then, we can construct a protocol between Alice and Bob where Alice aims to send $\vecx \in [N]^m$ to Bob using this algorithm. 
    We need appropriate subroutines for the $(t,w)$-QGGM.
    For the group operation gate, it works as follows.

    \paragraph{Subprotocol $\DGop$ for group operation gates.}
    The initial states are
    \[
        \sum_{b,i\in[t],j\in[w]}\beta_{b,i,j} \ket{b,i,j}_{\vecB\vecT\vecW}
        \otimes \sum_{z,w} \alpha_{z,w}\ket{...,g^z,...,g^w,...}_{ ...\vecX_i...\vecY_j...}
    \]
    where Alice holds the registers $\vecX=(\vecX_1,...,\vecX_t),\vecY=(\vecY_1,...,\vecY_w)$ and Bob holds $\vecB,\vecT,\vecW$. $g^z$ and $g^w$ are stored in $\vecX_i,\vecY_j$, respectively.

    \begin{enumerate}
        \item Bob sends $\vecB,\vecT,\vecW$ to Alice.
        Alice applies quantum group operation gates on $\vecB\vecT\vecW\vecX\vecY$.
        \item Alice returns the register $\vecB,\vecT,\vecW$ to Bob.
    \end{enumerate}
    In this protocol, Alice sends a quantum state of dimension $2tw.$
    The other element gates can be delegated analogously.
    
    Each quantum element gate is operated with a quantum state with $2tw$ dimension, thus \Cref{lem: interactive compression} implies that
    \[
        \frac{\log\epsilon' + m\log N }{2 } \le Q \log(2tw) ~\Longrightarrow~ \epsilon \le \frac{(2tw)^{2Q}}{N^m} + \frac{(C+m+1)^2}{2p},
    \]
    which concludes the proof.
    \ifnum\llncs=1 \qed \fi
\end{proof}

\subsection{Unknown-order QGGM}

We can prove the following QGGM variant of~\Cref{thm: uGGM-order,thm: uGGM-RSA}.
\begin{theorem}\label{thm: uQGGM-order}
    Let $\cA_\ord$ be a $(C,Q)$-algorithm to solve the order-finding problem over $\cD_\pprime^{(n)}$ in the $(t,w)$-QGGM. It holds that
    \[
        \Pr_{\cA_\ord,N} \left[
            \cA_\ord^{\ket{\cG_N}}(g) \rightarrow N
        \right]
        =O\left(
            \frac{C^3}{2^{n}} + \frac{n(2tw)^{2Q}}{2^n}
        \right).
    \]
    For the product-order-finding algorithm over $\cD_\pprime^{(n)}$ in the QGGM, it holds that
    \[
        \Pr_{\cA_\ord,p,q} \left[
            \cA_\ord^{\ket{\cG_{pq}}}(g) \rightarrow pq
        \right]
        =O\left(
            \frac{T^4}{2^{2n}} + \frac{n^2(2tw)^{2Q}}{2^{2n}}
        \right).
    \]
\end{theorem}
\begin{proof}[{\ifnum\llncs=0 Proof sketch\else Sketch\fi}]
    The proof is almost identical with the known-order QGGM proofs. 
    Instead of applying~\Cref{lem: classical_eq_remove}, whenever the classical preprocessing finds an informative collision, we use it to compress $N$. Otherwise, Bob can delegate quantum group operations to Alice to construct the interactive protocol encoding $N$. The product-order-finding case is analogous.
    \ifnum\llncs=1 \qed \fi
\end{proof}

\section{Lower Bounds in the Quantum Generic Ring Model}
\subsection{Quantum Generic Ring Model}
We define the quantum generic group model (QGRM) in this section.
The QGRM is a natural analog of the classical generic ring model~\cite{AM09,JS13}, similar to the relation between the QGGM and GGM. 

Let $\cR$ be a commutative ring isomorphic to $\Z_N$ for some integer $N$ to be specified later.
Let $t,w$ be positive integers.
A quantum generic ring algorithm $\cA$ in the QGRM is defined as follows. Note that the definition of ring multiplication and division is rather complicated because of their subtlety; for example, they are not invertible as is, or there is no inverse. Our abstraction closely resembles the actual target arithmetic gates of circuit optimizations, e.g.,~\cite{Bea03,Gid19}.
We also remark that the $(t,w)$-QGGM can be defined analogously.
\begin{itemize}
    \item There are two registers: qubit and element registers holding superpositions of some information. In contrast, element registers take a set of elements $x\in \cR \cup \{\bot\}$ as the computational basis.
    The algorithm arbitrarily appends a new element register initialized by $\ket{0}$ or $\ket{1}$.
    \item There are (arbitrary) quantum gates that map qubits to qubits, which cannot take element registers as input.
    \item There are special gates called \emph{element gates} that can access the element wires as follows. The unspecified registers are unchanged by the operations.
    \begin{description}
        \item[Ring Addition Gate.] It takes a qubit register $\vecB$ and two element registers $\vecX,\vecY$ and applies the unitary that works on the computational basis as follows:
        \[
            \ket{b}_{\vecB} \ket{x,y}_{\vecX\vecY} 
            \mapsto\ket{b}_{\vecB} \ket{x+by,y}_{\vecX\vecY}
        \]
        for $x,y \neq \bot$, otherwise do nothing.
        \item[Ring Subtraction Gate.] It is essentially identical to the ring addition gate, except for the choice of unitary:
        \[
            \ket{b}_{\vecB} \ket{x,y}_{\vecX\vecY} 
            \mapsto\ket{b}_{\vecB} \ket{x-by,y}_{\vecX\vecY}.
        \]
        \item[Ring Product-Addition Gate.] It takes a qubit register $\vecB$ and three element registers $\vecX,\vecY,\vecZ$ and applies the unitary that works on the computational basis as follows:
        \[
            \ket{b}_{\vecB} \ket{x,y,z}_{\vecX\vecY\vecZ} 
            \mapsto\ket{b}_{\vecB} \ket{x+byz,y,z}_{\vecX\vecY\vecZ}
        \]
        for $x,y,z \neq \bot$ and registers, otherwise do nothing.
        \item[Testing Invertible Gate.] It takes an element register $\vecX$ and a qubit register $\vecC$, and applies the unitary $U_\Test$ that works on the computational basis as follows:
        \[
            U_\Test:\ket{x,c}_{\vecX\vecC} \mapsto \ket{x,c\oplus \Test(x)}_{\vecX\vecC}
        \]
        where $\Test(x)=1$ if $x$ is invertible, otherwise $\Test(x)=0$.
        \item[Ring inversion-addition Gate.] It takes a qubit register $\vecB$ and three element registers $\vecX,\vecY,\vecZ$. It appends an ancillary qubit register $\vecC$ initialized by $\ket{0}$ and applies the following sequence of unitaries that works on the computational basis as follows:
        \begin{align*}
            &\ket{b}_{\vecB} \ket{x,y,z}_{\vecX\vecY\vecZ} \ket{0}_\vecC\\
            &\mapsto\ket{b}_{\vecB} \ket{x,y,z}_{\vecX_i,\vecY_j,\vecY_k}\ket{\Test(z)}_\vecC\\
            &\mapsto\ket{b}_{\vecB} \ket{x,+b\Test(z) y\cdot z^{-1},y,z}_{\vecX\vecY\vecZ}\ket{\Test(z)}_\vecC\\
            &\mapsto\ket{b}_{\vecB} \ket{x,+b\Test(z) y\cdot z^{-1},y,z}_{\vecX\vecY\vecZ}\ket{0}_\vecC
        \end{align*}
        for $x,y,z\neq \bot$, and do nothing other cases.
        Here, the first and last unitaries are $U_\Test$ on $\vecZ\vecC$. 
        In the second unitary, it adds $y \cdot z^{-1}$ only if $\Test(y_k)=1$ and $b=1$. It discards $\vecC$ in the end.
        \item[Equality Gate.] It takes a qubit register $\vecB$ and two element registers $\vecX,\vecY$ and applies the unitary that works on the computational basis as follows:
        \begin{align*}
                \ket{b}_{\vecB} \ket{x,y}_{\vecX\vecY} 
                \mapsto\ket{b\oplus\delta_{x,y}}_{\vecB} \ket{x,y}_{\vecX\vecY}
        \end{align*}
        for $x,y \neq \bot$, otherwise do nothing,
        where $\delta_{x,y}=1$ if $x=y$ and $0$ otherwise.
    \end{description}
    \item We allow the intermediate measurements for registers and the labeling gate.
    \begin{description}
        \item[Classical Labeling Gate.] It takes $\lceil \log_2 N\rceil$ qubit registers, measures it, and interprets them as an element in $x \in \Z_N \simeq \cR$. It appends a new element register holding $\ket{x}$. If there is no corresponding element $x \in \Z_N$ to the input wires, it outputs an element wire containing $\ket{\bot}$.
    \end{description}
\end{itemize}
We count the number of element gates as the cost metric, denoted by \emph{the ring operation complexity}.

We stress that the modulus $N$ is \emph{NOT} given to the generic algorithms explicitly. Instead, the algorithm only accesses the ring (or modular) operations. 
This still captures most quantum parts of the quantum factoring algorithms. For Shor's algorithm, the quantum part aims to find the order $r$ of the randomly chosen integer $a$, which only requires modular arithmetic. We prove the lower bound for the order finding in the QGRM in~\Cref{thm: QGRM_order}. 
The knowledge of $N$ beyond our model is used in the classical post-processing parts for computing (say) $\gcd(a^{r/2}+1,N)$.

Regev's algorithm does not compute the order. Instead, it computes a short vector $\vecz=(z_1,...,z_d)$ in a certain lattice, and then compute $\gcd(b_1^{z_1}...b_d^{z_d}-1,N)$. In the QGRM, the algorithm still can compute $b_1^{z_1}...b_d^{z_d}-1$ as an integer, which is relatively small, having a nontrivial common factor with $N$. We prove that this algorithm needs to make a logarithmic number of ring operations in~\Cref{thm: QGRM_fact}.

\subsubsection{Quantum Generic Ring Algorithm with Classical Preprocessing}
We consider a slightly more general algorithm that can do classical preprocessing \emph{without} the testing or equality gates. The classical ring operations are defined by the ring operations that element gates are measured before applying ring operations. 
Recall the ring operations are done on Alice's side in the proofs. The delegation of classical ring operations can be done without Alice's messages, except for the equality gates, where Bob needs to take the output of gates.
Therefore, the number of classical ring operations, such as precomputing $x^{2^k}$, is irrelevant to our lower bounds. We do not explicitly discuss this setting in the remainder of this section.

\subsection{Lower Bounds in the QGRM}
This section is devoted to proving that the order-finding problems with certain distributions and a certain type of factoring algorithms have logarithmic ring operation complexity.
In a ring $\cR\simeq \Z_N$, the order of $x\in \cR$, denoted by $\ord_N(x)$, is defined by the minimal positive integer $e$ such that $x^e = 1 \bmod N$.
A prime number $p$ is \emph{safe} if $\frac{p-1}{2}$ is also prime. 
We consider the following problem. 
\begin{problem}
    Let $\cD_\safe^{(n)}$ be a uniform distribution over the set of $n$-bit safe primes.
    An order-finding problem over $\cD_\safe^{(n)}$ in the QGRM is defined as follows. First, two distinct primes $p,q$ are sampled from $\cD_\safe^{(n)}$ and let $N=pq$. Choose a random $x\in\Z_N$.
    The adversary is given $x$ stored in the element register in the QGRM for $\cR \simeq \Z_N$ and asked to find $\ord_N(x)$.
\end{problem}
We assume that the number of $n$-bit safe primes is at least $C\cdot 2^n/n^2$ for some constant $C>0$, which is a variant of the conjecture that the number of safe primes below $N$ is of order $\Theta(N/\log^2 N)$~\cite[Section 5.5.5]{Shoup09}. 

\begin{theorem}\label{thm: QGRM_order}
    Let $\cA_\ord$ be an order-finding algorithm over $\cD_\safe^{(n)}$ in the $(t,w)$-QGRM with the ring operation complexity of $Q$. Assuming that the number of $n$-bit safe primes is at least $C\cdot 2^n/n^2$ for $C>0$, it holds that
    \[
        \Pr_{\cA_\ord,p,q,x}\left[
            \cA_\ord(x) \rightarrow \ord_{pq}(x)
        \right]
        = O \left(
            \frac{n^4 (2tw)^{2Q}}{2^{2n}} + \frac{1}{2^n}
        \right)
    \]
\end{theorem}
\begin{proof}
    Let $N=pq$.
    We first observe that the order of $x$ is a divisor of $\frac{(p-1)(q-1)}{2} = 2 \cdot \frac{p-1}{2} \cdot \frac{q-1}{2}.$ With probability $1-O(1/p)$ over random $x$, $\ord_N(x) = \frac{(p-1)(q-1)}{2} $ or $\frac{(p-1)(q-1)}{4} $.
    In this case, one can recover $(p,q)$ for safe primes $p,q$ from $\ord_N(x)$ using the factorization. 
    
    Based on this observation, we construct a protocol between Alice and Bob where Alice wants to send $(p,q)$ to Bob. The proof is identical to that of~\Cref{thm: QGGM_DL}, except that we need the delegation sub-protocols for ring operations. By the assumption, Alice sends one out of $O\left(\frac{2^{2n}}{n^4}\right)$ candidates to Bob using $Q$ qubits of communications.
    \ifnum\llncs=1 \qed \fi
\end{proof}

We then consider the factoring algorithms, where the generic algorithm's goal is to find an integer with a nontrivial common divisor with $N$. 
We prove the following theorem.
\begin{theorem}\label{thm: QGRM_fact}
    Let $\cA$ be an algorithm in the $(t,w)$-QGRM with the ring operation complexity of $Q$. 
    For two primes $p,q$ sampled from $\cD_\pprime^{(n)}$ and $N=pq$, it holds that
    \[
        \Pr_{\cA,p,q}\left[
            1<\gcd(Z,N)<N :\cA() \rightarrow Z
        \right]
        = O \left(
            \frac{n \log Z (2tw)^{2Q}}{2^{n}}
        \right).
    \]
    In particular, if $\log Z = O(2^{(2-\epsilon)n})$ for any $\epsilon>0$, this implies that $Q=\Omega(\frac{\log N}{\log(2tw)})$ to have the constant success probability.
\end{theorem}
Before proceeding with the proof, we give some interpretations of this theorem.
As we do \emph{not} give $N$ to the generic algorithm, it cannot apply the modulus operation. Therefore, the known quantum factoring algorithms must be explained with some modifications, where the final steps usually compute the common divisor of some integer and $N$.

Instead of giving $N$, we ask to find an integer that suffices for factoring $N$. In the QGRM, this integer must be computed in plain, without modulus computation. 
For Regev's algorithm, the final integer is of the form $Z=\prod_{i\in [d]} b_i^{z_i}$ for $z_i = \exp(O(\sqrt n))$ and $d\approx \sqrt n$. The last statement holds in this case as well.

The output of Shor's algorithm corresponds to $Z=a^{r/2}-1$ for $r=\ord_N(a)$. 
The bit length of $Z$ is about $\log Z \le \frac{r}{2} \cdot \log a \le \frac{nN}{2}$. Therefore, we cannot apply this theorem to Shor's algorithm in general.

\begin{proof}
    Let $\epsilon$ be the success probability of $\cA$.
    We construct a protocol that sends $(p,q)$ using $\cA$. Precisely, Bob runs $\cA$ using the delegation of quantum ring operations using $Q \log (2tw)$ qubits. After obtaining the outputs $Z$ from $\cA$, Alice additionally sends an index of the prime factor of $Z$ among its $n$-bit prime factors. Since the number of $n$-bit primes factors of $Z$ is bounded by $\log Z/n$, 
    the index can be described in $\log(\log Z)-\log n$ classical bits.
    Finally, Alice sends the other prime factor which can be specified by $n-\log n +O(1)$ classical bits.
    Appying~\Cref{lem: interactive compression}, we have
    \[
        2 Q \log (2tw) +( \log \log Z - \log n )+(n-\log n +O(1))\ge 2n-2\log n + \log \epsilon + O(1),
    \]
    which implies
    \[
        \epsilon = O \left(\frac{n(\log Z)(2tw)^{2Q}}{2^{n}}\right),
    \]
    concluding the proof.
    \ifnum\llncs=1 \qed \fi
\end{proof}

\section{Lower Bounds for Index Calculus Algorithms}\label{sec:smoothGGM}
This section introduces a new model called \emph{the smooth index calculus model (SGGM)} of generic algorithms, including the (simplest) index calculus methods.

A main feature of index calculus is using the set of $B$-smooth numbers, denoted by $S_B$, whose prime factors are all less than or equal to $B$. These numbers are relatively quickly factorized, and the index calculus method finds many nontrivial elements in $S_B$ to leverage this fact. 

\subsection{Smooth Generic Group Model}
The smooth GGM is parameterized by a parameter $B$, which induces the factor base $\cB$ and the set $S=\{h_1,...,h_{|S|}\}$ of smooth elements.
Precisely, the factor base is a set of primes $\cB = \{p_1,...,p_b\}$ and the smooth element $h_i\in S$ is of the form
\begin{equation}\label{eqn:smoothing}
    h_i  = p_1^{c_1^{(i)}} \cdot \cdots \cdot p_b^{c_b^{(i)}}
\end{equation}
for $\vecc^{(i)}=(c_1^{(i)},...,c_b^{(i)})$, whose precise conditions will be specified later.

An SGGM algorithm $\cA$ over $\cG$ of order $N$ with the parameter $B$, denoted by an algorithm in the $B$-SGGM, is given by a circuit with the following features:
\begin{itemize}
    \item There are two types of wires: bit wires and element wires. Bit wires take a bit in $\{0,1\}$, and element wires take an element in $x \in \Z_N \cup \{\bot\}$, which is denoted by $g^x$.
    \item There are bit gates and element gates that are identically defined as the generic group model. 
    \item 
    There are special element gates defined as follows:
    \begin{description}
        \item[Smooth Test Gate.] It takes an element wire containing $h$. If $h \in S$, it outputs $1$, otherwise outputs $0$. 
        \item[Smoothing Gate.] It takes an element wire containing $h$. If it is smooth, i.e., $h=h_i$ for some $i\in [|S|]$, outputs $\vecc^{(i)}$ defined in~\Cref{eqn:smoothing}. Otherwise, it outputs $\bot.$
    \end{description}
\end{itemize}
We further establish the properties of the factor base and the smooth elements regarding the parameter $B$.
Let $g$ be the generator of $\cG$, and let $u>0$ be such that $B=N^{1/u}$. Let $c_\base,c_\smooth,d_\smooth\ge 1$ be the universal constants that are independent from $B$. Here, $o(1)$ hides a factor much less than $1$.
\begin{itemize}
    \item The set $\cB$ is given to the algorithm. The set $S$ is randomly chosen and unknown to the algorithm.
    For the factor base $\cB=\{p_1,...,p_b\}$, it holds that $p_i=g^{z_i}$ for some \emph{random} $z_i$ for each $i$, which is unknown to the algorithm.
    \item The size of the factor base $|\cB|=b$ is $(c_\base+o(1))  B/\log B$.
    \item The number of smooth elements is
    \begin{equation}\label{eqn: esti_smooth}
        p_S:=\frac{|S|}N = 
        (c_\smooth+o(1)) \cdot \left(
            \frac{d_\smooth+o(1)}{u\log u}
        \right)^u.
    \end{equation}
    \item For any rank-$c$ affine space $V$ in $\Z_N^b$, define
    \[
        S_V:=\{
            h^{(i)}\in S: \vecc^{(i)} \in V
        \}.
    \]
    If $c=(c_\base+o(1)) C / \log C$ for some $C=N^{1/v}$, it holds that
    \begin{equation}\label{eqn:rank-smoothing}
        \frac{|S_V|}N \le 
        (c_\smooth+o(1)) \cdot \left(
            \frac{d_\smooth+o(1)}{v\log v}
        \right)^v.
    \end{equation}
\end{itemize}
We explain the reasoning behind these assumptions. 
The assumption on the prior knowledge of the algorithm reflects the reality.
The randomness of $S$ and $z_i$ prevents the generic algorithm from using the explicit values related to the smooth elements.

The sizes of $\cB$ and $S$ stem from the original choices in the index calculus, whose estimated sizes are well-studied. We refer the survey on this topic~\cite{Gra08} to the readers.

The last assumption describes that the vectors $\vecc^{(i)}$ are \emph{well-distributed}. In particular, it asserts that the factor base $\{p\le C\}$, which corresponds to the subspace $V=\Z_N^c \times \{0\}^{b-c}$, maximizes the size of $S_V$, according to the estimated size by~\Cref{eqn: esti_smooth}.

\subsubsection{Polynomial Representations}
Again, we identify the group elements by the corresponding polynomials. We mainly focus on the discrete logarithm problems where the problem instance is given as $(g,g^{x_1},...,g^{x_m})$, which corresponds to $1,X_1,...,X_m$. Furthermore, because of the factor base, we have more formal variables $Z_1,...,Z_b$. Therefore, each element corresponds to the polynomial in
\[
    \Z_N[X_1,...,X_m,Z_1,...,Z_b].
\]
We stress that we occasionally identify the polynomial with its coefficient vector.

We consider the answer from the smoothing gate to be the collision. Precisely, if an element $h$ corresponding to the polynomial $P$ is given to the smoothing gate and the answer is $(c_1,...,c_b)$, then it induces the collision
\[
    P = c_1 Z_1 + ... + c_b Z_b.
\]
If it is not included in the span of the previous zero set $\cZ$, we include it as an informative collision as well.

\subsection{The Discrete Logarithm Problem in the SGGM}
In this section, we assume that the variables $N=N(\lambda),B=B(\lambda),u=u(\lambda)$ are parameterized by some implicit parameter $\lambda$ so that we can work in the asymptotic regime. Still, we drop the parameter $\lambda$ for simplicity.

\begin{theorem}\label{thm: SGGM-DL}
    Let $\cG$ be a cyclic group of prime order $N$.
    Let $B$ be an integer such that $B=N^{1/u}$ for some $u>0$.    
    Let $\cA_\dl$ be a DL algorithm in the $B$-SGGM with a constant success probability. Then, the number of group operations $T$ of $\cA_\dl$ must satisfy
    \[
    T = \exp\left(
        \Omega\left(\sqrt{\log N \log\log N}\right)
    \right).
    \]
\end{theorem}
\begin{proof}
    Toward contradiction, we assume that $\cA_\dl$ successfully solves the DL problem in the SGGM with smaller group operations than the statement.
    We first observe that each equality query makes an informative collision with probability $\frac1N$; thus, with probability $1-\frac{T^2}{N}$, there are no informative collisions from the equality queries. From now on, we ignore the equality gates and assume that all informative collisions are from the smoothing gate.
    
    As before, we assume that $\cA_\dl$ makes the equality gate at the end so that the collision is found. 
    We begin with the following fact, which is a SGGM variant of~\Cref{lem: SZinformative}.
    \begin{fact}\label{fact: SGGM_ps}
        Each group operation introduces a new informative collision (through the smoothing gate) with probability at most $p_S$ defined in~\Cref{eqn: esti_smooth}. In particular, the input element to the smoothing gate collides with a random element in $S$.
    \end{fact}
    \begin{proof}[{\ifnum\llncs=0 Proof \fi}of fact]
        Let $h$ be a new group element corresponding to $(c_0,a,c_1,...,c_b)$, which is linearly independent from the vectors in $\cZ.$ This means that $h$ is uniformly distributed over random $X,Z_1,...,Z_b$ conditioned on the equations in $\cZ$ hold. That is, $h \in S$ holds with probability $|S|/N=p_S.$
        \ifnum\llncs=1 \qed \fi
    \end{proof}

    \paragraph{Case 1.}
    We first consider the case that $u$ is sufficiently large so that
    \[
        u^{8u^2} \ge N ~\Longrightarrow~
        u\log u =\Omega\left ( \sqrt{\log N \log\log N} \right).
    \]
    In this case, by~\Cref{fact: SGGM_ps}, $\cA_\dl$ must make 
    \[1/p_S = \Omega( u^u) = \exp \left( \Omega\left ( \sqrt{\log N \log\log N} \right)\right)
    \]
    group operations to find an informative collision with a constant probability.\footnote{A formal proof requires some probabilistic arguments, which we omitted here.}

    \paragraph{Case 2.}
    We consider the other case that $u$ is relatively small so that
    \[
        u^{8u^2} \le N.
    \]
    In this case, we choose $v>u$ such that $v^{v^2}=N$.
    Note that
    \[
    (2u)^{(2u)^2} = (2u)^{4u^2} \le u^{8u^2} \le N,
    \]
    thus $v \ge 2u$ holds.
    Suppose that the algorithm finds $K$ informative collisions at total.
    We will prove that $K=\Omega\left( N^{1/2v}\right)$ in this case. Since
    \[
    v^2 \log v = \log N ~ \Longrightarrow v = \Theta\left ( \sqrt{\frac{\log N}{\log\log N}}\right),
    \]
    we have
    \[
        T \ge K=\exp \left(\Omega\left( \frac{\log N}{v}\right)\right) =\exp \left( \Omega\left ( \sqrt{\log N \log\log N} \right)\right). 
    \]
    Combining the two cases, we prove the theorem.

    It remains to prove the lower bound of $K$ in the second case.
    We identify the formal variable $X$ to represent $g^x$. 
    In particular, the span of the final zero set $\cZ$ must include the polynomial $X-x$, or a vector $(-x,1,0,...,0).$ 
    We define $\cZ^{(t)}$ to denote the zero set right after the $t$-th informative collision.
    Define the following projections of $\cZ^{(t)}$:
    \begin{align*}
        \cZ_{X=x}^{(t)}&:= \{ 
            (ax+c_0,c_1,...,c_b):
            c_0+aX+c_1Z_1+...+c_bZ_b \in \cZ^{(t)}
        \}\\
        \cZ_\cB^{(t)}&:= \{ 
            (c_1,...,c_b):
            c_0+aX+c_1Z_1+...+c_bZ_b \in \cZ^{(t)}
        \}
    \end{align*}

    We observe the following facts.
    \begin{fact}\label{fact: SGGM_projspan}
        The rank of $\cZ^{(K)}_{X=x}$ is less than $K$.
        In particular, there must exist a smoothing gate making the $t(\le K)$-th informative collision $P$ such that $P_{X=x}^{(t)}$ is included in the span of $\cZ_{X=x}^{(t-1)}.$
        The rank of $\cZ^{(t)}_{X=x}$ is equal to the rank of $\cZ^{(t)}_\cB$ for each $t\in[K].$
    \end{fact}
    \begin{proof}[{\ifnum\llncs=0 Proof \fi}of fact]
        Let $(-x,1,0,...,0)=\vecb$ and $\{\vecb,\vecb_2,...,\vecb_K\}$ be the basis extension of $\cZ=\cZ^{(K)}$ from $\{\vecb\}.$ The projection $\pi:(c_0,a,c_1,...,c_b) \mapsto (c_0+ax,c_1,...,c_b)$ maps $\cZ$ to $\cZ_\cB^{(K)}$ and $\pi(\vecb)=0$, thus the rank of  $\cZ_\cB^{(K)}$ must be $K-1$.
        The final statement follows from $(1,0,...,0)$ is not included in the span of $\cZ.$
        \ifnum\llncs=1 \qed \fi
    \end{proof}
    We call the first smoothing gate by \emph{critical} with input $h$ and output $h^{(i)}\in S$ satisfying the condition described in~\Cref{fact: SGGM_projspan}.
    Let $(h_1,...,h_b) $ and $ \vecc^{(i)} $ be the corresponding coefficient vectors of $h$ and $h^{(i)}$.

    We give an upper bound for the probability $p_t$ that the $t$-th informative collision is critical.
    This means that the rank of $\cZ_\cB^{(t-1)}$ is $t-1$, and $(h_1,...,h_b) - \vecc^{(i)}$ is included in $\cZ_\cB^{(t-1)}$. In other words, 
    \[
        \vecc^{(i)} \in (h_1,...,h_b) + \spann \left(\cZ_\cB^{(t-1)}\right) =: V.
    \]
    Since $h^{(i)}$ is a random element in $S$,~\Cref{eqn:rank-smoothing} implies that the probability $p_t$ is bounded by
    \begin{align*}
         p_t =  \Pr \left[
            \vecc^{(i)} \in V
        \right]
        =  \frac{|S_V|}{|S|} .
    \end{align*}
    Let $C=N^{1/v}$ and $c= c_\base C/\log C$. If $t \le c$, the logarithm of the above equation becomes for a constant $\alpha\approx \log d_\smooth$
    \begin{align*}
        &\log \left(\frac{|S_V|}{|S|}  \right) \le - v \log (v\log v) +u \log (u\log u) +\alpha (v-u) + O(1)\\
        &\le - 0.5 v \log (v\log v) 
        -u \log (2u \log 2u)  +u \log (u\log u) +\alpha v+O(1)\\
        &\le -0.5v\log v + O(1).
    \end{align*}
    where we use the fact that $d_\smooth$ is constant and $v \ge 2u$, and set $\alpha \approx \log d_\smooth $, which is less than $0.5 +0.5\log\log v$ in the interested parameter regime. 
    It implies that $ p_t \le {\beta}/{v^{0.5v}}$ for some constant $\beta>0$.
    In other words, with constant probability, the critical informative collision will be found after $\Omega(v^{0.5v}) = \Omega(N^{1/2v})$ informative collisions are found.
    \ifnum\llncs=1 \qed \fi
\end{proof}

\bibliographystyle{alpha}
\bibliography{ref}
\appendix
\section{Missing Proofs}
\subsection{Missing proofs in the GGM}\label{missing_GGM}

\begin{proof}[{\ifnum\llncs=0 Proof \fi}of~\Cref{lem:simulation_without_queries}]
    Suppose that $\cA$ is deterministic; if $\cA$ is randomized, we include the random seed as its description. We construct an algorithm $\cA'$ that has the same circuits as $\cA$, but the element wires are replaced by the polynomial wires.     
    The initialization $\cP=\{(w_1,P_1),...,(w_m,P_m)\}$ for input can be done without any query, and each input element wire $w_i$ is replaced by $P_i$ in $\cA'$. The labeling gates and group operation gates are processed as in the polynomial list. For the equality gates with element wires $w_i,w_j$, the output can be computed by checking if $P_i-P_j$ is included in the span of $\cZ$; if included the output of the equality gate is $1$, otherwise $0$.\ifnum\llncs=1 \qed \fi
\end{proof}

    \begin{proof}[{\ifnum\llncs=0 Proof \fi}of~\Cref{fact:MDL}]
        Let $P_i := a_{i,1}X_1+...+a_{i,m}X_m + b_i$. 
        A collision $(i,j)$ induces a linear equation over $\Z_p$ as 
        \[
            0 = P_i - P_j = (a_{i,1}-a_{j,1}) X_1 + ... + (a_{i,m}-a_{j,m}) X_m + (b_i-b_j).
        \]
        Since collisions are informative, they are nontrivial and linearly independent due to~\Cref{eqn:GGM_equalgate_zeropoly}. Thus $m$ informative collisions give a system of $m$ linear equations over $\Z_p$ with $m$ variables that are linearly independent, which can be easily solvable.\ifnum\llncs=1 \qed \fi
    \end{proof}

\begin{proof}[{\ifnum\llncs=0 Proof \fi}of~\Cref{thm: GGM_GapCDH}, sketch]
    Observe that without finding an informative collision, the output should correspond to $aX+bY+c$ for some $a,b,c$, where $X,Y$ are the variables corresponding to $g^x,g^y$. The probability that $aX+bY+c=XY$ is at most $1/|\cG|$ over the random choice of $X,Y$. Therefore, the algorithm must find an informative collision.

    Given an informative collision exists, the first part of the encoding is the first informative collision. If this collision has a nonzero coefficient for the monomial containing $X$, then the second part of the encoding is $x$. Otherwise, it is $y$. Given the encoding, the decoding procedure is 1) parses the first informative collision and one of $x$ or $y$, and 2) plugs it in the first collision. The collision collapses to a one-variable polynomial of degree less than 1, and the correct solution can be guessed with probability at least $1/2$.
    \ifnum\llncs=1 \qed \fi
\end{proof}

\subsection{A QGGM lemma}\label{missing_QGGM}

\begin{proof}[{\ifnum\llncs=0 Proof \fi}of~\Cref{lem: classical_eq_remove}]
    Except for the equality gates, we define the algorithm $\cA_c'$ as identical to $\cA_c$. It removes all equality gates, except the trivial equality gate (as in the classical GGM) that are replaced by bit flipping.
    Given the first assertion, the ``In particular'' part is obvious because the trace distance between the intermediate outputs of two $(C,Q)$-algorithms are identical with probability $1-\frac{(C+m+1)^2}{2p}$, and the remaining parts are the same.

    The proof proceeds as follows. As the algorithm $\cA_c$ is only given classical group elements and can apply classical group gates, it only maintains at most $C+m+1$ classical group elements, which are represented by polynomials $P_1,...,P_{C+m+1}$ as done in the classical generic group models. Each equality query corresponds to the difference between polynomials $P_i-P_j$.
    If $P_i-P_j$ is identically zero or never be zero, then there is no difference between $\cA_c$ and $\cA_c'$ from these equality gates.

    Consider an equality gate corresponding to $P_i-P_j$ that is not identically zero. There exists some prime power $q^t$ that exactly divides $N$ such that $P_i-P_j$ is nonzero modulo $q^t$. Since $P_i-P_j$ is linear, the portion of inputs where the equality gate corresponding to $P_i-P_j$ behaves differently from the identity gate is at most $1/q \le 1/p $. Since there are at most $\binom{C+m+1}{2} \le \frac{(C+m+1)^2}{2}$ different pairs of group elements, at most $ \frac{(C+m+1)^2}{2p}$-fraction of inputs make difference on the behaviors of $\cA_c$ and $\cA_c'$. In other words, the output states of the two algorithms are identical with probability at least $1-\frac{(C+m+1)^2}{2p}$ for random inputs.
    \ifnum\llncs=1 \qed \fi
\end{proof}
\section{An Alternative Proof for the MDL Lower Bound}\label{app:MDL}
We give a simple proof for the MDL lower bound~\Cref{thm:MDL_GGM}.
We begin with the proof of~\Cref{lem: SZinformative}.

\begin{proof}[{\ifnum\llncs=0 Proof \fi}of~\Cref{lem: SZinformative}]
    Let $\cZ=\{Q_1,...,Q_s\}$ be the current zero set. Assume that $s<t$; otherwise, there is no more informative collision.
    Let $P$ be the linear polynomial corresponding to the new collision.
    Assume that $P \notin \spann(\cZ).$
    This implies that $P$ is nonzero in the quotient ring $\Z_N[X_1,...,X_t]/\spann(\cZ)\simeq \Z_N[L_1,...,L_{t-s}]$ for some linear polynomials $L_1,...,L_{t-s}$, and each variable $L_i$ is uniform random over random choice of $x_1,...,x_t$ conditioned on $Q_1,...,Q_s=0$, making $P$ uniform over $\Z_N$. That is, $P=0$ holds and is informative with probability $1/N$.
    \ifnum\llncs=1 \qed \fi
\end{proof}
For the readability, we restate the lower bound.
\begin{theorem}\label{thm:APP_MDL}
    Let $\cG$ be a cyclic group of prime order.
    Let $\cA_{m\text{-}\mdl}$ be an $m$-MDL algorithm in the GGM having at most $T$ group operation gates. It holds that:
    \[
        \Pr_{\cA_{m\text{-}\mdl},\vecx}\left[
            \cA_{m\text{-}\mdl}^{\cG}(g,g^{\vecx}) \rightarrow \vecx
        \right]
        =
        O\left( \left( \frac{e(T+2m)^2}{2m|\cG|}\right)^m \right).
    \]
\end{theorem}
\begin{proof}
    As seen in the original proof of~\Cref{thm:MDL_GGM}, we can assume that the algorithm finds $m$ informative collisions to solve the $m$-MDL problem. Let $T$ be the number of group operations. 
    
    By~\Cref{lem: SZinformative}, each equality gate induces an informative collision with probability at most $1/N$. 
    We further assume that the algorithm never applies the equality gates to the predictable inputs. This makes the probability that each equality gate is informative equal to $1/N$ independent from the previous equality gates.
    
    Let $E$ be the number of equality gates, which is at most $\binom{T+2m}{2}\le \frac{(T+2m)^2}2$.
    Assume that $\binom{T+2m}{2}\le mN,$ otherwise the upper bound becomes larger than 1.
    Let $C$ be the number of informative collisions during the algorithm and $\mu=\E[C]=\frac{E}{N}.$
    Let $\delta=\frac{mN}{E}-1.$
    Note that $\delta\mu \le (1+\delta) \mu =m.$
    By the multiplicative Chernoff bound, we have
    \begin{align*}
        \Pr[C \ge m] &\le \left(
            \frac{e^\delta}{(1+\delta)^{1+\delta}}
        \right)^\mu
        = 
            \frac{e^{\mu\delta}}{(mN/E)^{m}}
        \\
        &\le
        \frac{e^{m}}{(mN/E)^m} = \left(
        \frac{eE}{mN}
        \right)^m\\
        &\le\left(
        \frac{eE}{mN}
        \right)^m \le \left(
        \frac{e(m+2T)^2}{2mN}
        \right)^m.
    \end{align*}
    Since this is an upper bound of the success probability of the $m$-MDL algorithm, it concludes the proof.
    \ifnum\llncs=1 \qed \fi
\end{proof}
\section{Equivalence between GGMs}\label{sec:equiv}
This section proves that any single-stage problems secure in the Maurer-style (or type-safe) generic group model are also secure in the Shoup-style (or random representation) generic group model. The proof is essentially the same as~\cite[Theorem 3.5]{Zhandry22a} with some additional finer analysis.\footnote{In the original paper, the author only considers the polynomially-bounded algorithms and negligible advantage. We need to consider more fine-grained equivalence for the exact advantage and any number of group operations.}

We call the generic algorithms described in~\Cref{sec:classical} by type-safe (TS). We consider another style of generic algorithm that is called random representation (RR) introduced in~\cite{Shoup97}. In this model, a set $S\in\bit^*$ (with the known maximal length of elements) is given public, and a random injection $L:\Z_N \rightarrow S$ is chosen, which is called by the \emph{labeling function}. $L(x)$ is understood as a group element $g^x$. A generic algorithm in the random representation model is able to make the following queries:
\begin{description}
        \item[Labeling Query.] It takes $x \in \Z_N$ as input and outputs $L(x)$.
        \item[Group Operation Query.] It takes $\ell_1, \ell_2 \in S$ and a single bit $b$ as input. If there exist $x_1,x_2 \in \Z_N$ such that $L(x_1)=\ell_1$ and $L(x_2) = \ell_2$, it outputs $L(x_1 +b x_2)$. Otherwise, it outputs $\bot.$
\end{description}
We count the number of queries as a unit cost. A generic algorithm in this model is denoted by $\cA^{\cG_{RR}}.$ Note that there is no equality query in this model, which can be done by comparing the labels without accessing the oracle. If an algorithm only makes queries with the inputs that it received before by some queries or input, then we call it \emph{faithful}. 

The following lemma shows that, when considering a single-stage game as in this paper, a faithful generic algorithm in the RR model is essentially the same as one in the TS model, but there is a subtle difference otherwise. We write $L(\vecx) = (L(x_1),...,L(x_m))$ for $\vecx= (x_1,...,x_m) \in \Z_N^m.$

\begin{theorem}\label{thm: equiv}
    Let $f$ be a function that takes an element in $\Z_N^m$ as input, $p>0$, and $D$ a distribution over $\Z_N^m$. Suppose for any generic algorithm $\cB$ in the TS model with $T+\Lambda$ group operation complexity, it holds that  
    \[
        \Pr_{\cB,\vecx \gets D}[\cB^{\cG}(g^{\vecx},{\sf aux})=f(\vecx)] \le p
    \]
    where ${\sf aux}$ is a bit string, $\Lambda$ is to be specified and suppose that $f$ includes $k$ group element wires.
    
    Then, in the RR model, the following inequality holds
    \[
        \Pr_{\cA,L,\vecx \gets D}[\cA^{\cG_{RR}}(L(\vecx),{\sf aux})=f(\vecx)]\le p - \Delta
    \]
    where
    \begin{itemize}
        \item for a faithful generic algorithm $\cA$ with $T$ queries and $\Lambda=\Delta=0$, and
        \item in general, for a generic algorithm $\cA$ with $T$ queries such that at most $t$ labels that are not given to $\cA$ before, where $\Lambda=tr$ and $\Delta=r \cdot \left(\frac{T}{N}\right)^r$ for any positive integer $r$.
    \end{itemize}
    In particular, when $T=N^{1-1/c}$ for some constant $c>0$ and $p\ge 1/N$, we can choose $r=2c$, which asserts that the asymptotic results equally hold.
\end{theorem}
\begin{proof}
    Toward contradiction, we assume that there exists a generic algorithm $\cA$ in the RR model with the winning probability larger than $p-\Delta$. 
    We first consider the case that $\cA$ is faithful. 
    In this case, the algorithm $\cB$ proceeds as follows. $\cB$ initializes an empty table $T$, which will contain pairs $(h,\ell)$ for $h$ in an element wire and $\ell \in S$. This will be interpreted as $L(x)=\ell$. We define the following subroutines of $\cB$:
    \begin{description}
        \item[$\FindLabel(h)$:] It takes an element wire containing $h$ as input. It searches for a pair $(h',\ell) \in T$ with $h=h'$ using the equality gates. If such a pair exists, it returns $\ell$. Otherwise, it samples a random $\ell \in S$ conditioned on $\ell$ not being in the table $T$. It adds $(h,\ell)$ to $T$ and returns $\ell$.
        \item[$\FindElt(\ell)$:] It searches for a pair $(h,\ell') \in T$ with $\ell=\ell'$. If such a pair exists, it returns $h$ on an element wire. Otherwise, it generates an element wire containing $\bot$ and adds $(\bot,\ell)$ to $\ell$. It returns $\bot$. (This case does not occur for the faithful algorithms.)
    \end{description}
    $\cB$ executes $\cA$ and processes the queries from $\cA$ and the inputs/outputs as follows.
    \begin{itemize}
        \item Given the problem instance, $\cB$ parses it into a list $L$ of element wires. For each element wire $h\in L$, $\cB$ runs $\ell\gets \FindLabel(h)$ and sends $\ell$ to $\cA$ as a part of input corresponding to $h$.
        \item For a labeling query $x$ from $\cA$, $\cB$ constructs an element wire containing $g^x$ using a labeling gate. Then it runs $\ell \gets \FindLabel(g^x)$ and returns $\ell$ to $\cA$.
        \item For a group operation query $(\ell_1,\ell_2,b)$, $\cB$ runs $h_1 \gets \FindElt(\ell_1), h_2 \gets \FindElt(\ell_2)$, and computes $h=h_1 \cdot h_2^b$ using a group operation gate. Then it runs $\ell\gets \FindLabel(h)$ and returns $\ell$ to $\cA$.
        \item The final output of $\cB$ is identical to that of $\cA$. Precisely, if $\cA$ outputs $(\ell_1,...,\ell_k,\tau)$ for labels $\ell_1,...,\ell_k$ and a string $\tau$, $\cB$ runs $h_i\gets \FindElt(\ell_i)$ for $i\in [k]$ and outputs $(h_1,...,h_k,\tau)$.
    \end{itemize}

    Note that each labeling query and group operation query incurs a single element gate, thus the group operation complexity of $\cB$ is the same as one of $\cA$.
    To prove that $\cB$ wins with probability at least $p$, we consider the following sequence of hybrid experiments.

    \begin{description}
        \item[$H_0$.] In this hybrid, $\cA$ interacts with the group oracle $\cG_{RR}$. $\cA$ wins with probability at least $p$ by the assumption.
        \item[$H_1$.] This hybrid is the same as $H_0$ except that the random injection $L$ is lazily sampled. This is possible because $\cA$ is faithful. In the perspective of $\cA$, this is identical to $H_0$, thus the winning probability is the same as $H_0$.
        \item[$H_2$.] Here, $\cA$ is a subroutine of $\cB$. The view of $\cA$ is identical to that of $H_1$, and the translation between two models is done inside of $\cB$. The winning probability of $\cB$ is equal to that of $\cA$, which is the same as in $H_1$.
    \end{description}
    This completes the proof for the faithful $\cA$.

    We then consider the general case. In this case, $\cA$ may ask queries with the labels it never received. To remedy this, we need to modify the subroutine $\FindElt$, taking the probability that such a label is valid (i.e., an image of $L$) into account. The modified subroutine is as follows.
    \begin{description}
        \item[$\FindElt'(\ell)$:] It searches for a pair $(h,\ell') \in T$ with $\ell=\ell'$. If such a pair exists, it returns $h$ on an element wire. Otherwise, let $m:= | \{(h,\ell) \in T: h \neq \bot \}| $, and
        it does the following:
        \begin{itemize}
            \item With probability $1-(N-m)/(|S|-|T|)$, it generates an element wire containing $\bot$ and adds $(\bot,\ell)$ to $\ell$. It returns $\bot$.
            \item With probability $(N-m)/(|S|-|T|)$, it does the following procedures $t$ times: It randomly samples $x \in \Z_N$ and construct the corresponding element wire containing $g^x$, and searches for $(h,\ell') \in T$ with $h=g^x$ using the equality gates. If such a pair does not exist, it adds $(g^x,\ell)$ to $T$, returns $g^x$ on an element wire and halts. Otherwise, it discards $g^x$, and samples a fresh $x\in \Z_N$ and repeats.
        \end{itemize} 
    \end{description}
    A single iteration of the second case of $\FindElt'$ terminates with probability at least $1-m/N \ge 1-T/N$, and takes one labeling gate.

    In this case, the algorithm $\cB$ in the TS model is defined with $\FindElt'$ instead of $\FindElt.$ 
    This change makes $\cB$ find the corresponding element to the label that is not previously given.
    The success probability computation is almost identical, but incurs $(2T+k)\cdot \left(\frac{T}{N}\right)^{r}$ errors in the success probability regarding the failure of $\FindElt'$.\footnote{If we assume that $T\le N^{1-1/c}$ for some constant $c>0$, then repeating $r=2c$ times ensures that the probability of failure is $1/N^2$ for each $\FindElt'.$}
    If we carefully count the number $t$ of the labels that are not given before, the number of gates becomes $q+tr$ and $\Delta=r\cdot \left(\frac{T}{N}\right)^{r}.$

    \ifnum\llncs=1 \qed \fi
\end{proof}

\subsection{Lower bounds in the Random Representation GGM}\label{subsec: RRunknown}
Let $L:\Z_N\rightarrow S$ be the labeling function, which will be lazily sampled. 
We prove the RR GGM variant of~\Cref{thm: uGGM-order} in this section.

Note that the proof of~\Cref{thm: equiv} for the faithful case works well for the unknown-order group case. In other words, it suffices to focus on the algorithm's behavior to look for a new label that was not given to the algorithm before.
\begin{theorem}
    Let $\cA_\ord$ be an order-finding algorithm over $\cD_\pprime^{(n)}$ in the random-representation GGM with the group operation complexity $T$. It holds that
    \[
        \Pr_{\cA_\ord,N,L} \left[
            \cA_\ord^{\cG_N}() \rightarrow N
        \right]
        =O\left(
            \frac{T^3}{2^{n}}
        \right).
    \]
\end{theorem}
\begin{proof}[{\ifnum\llncs=0 Proof sketch\else Sketch\fi}]
    Suppose that $t$ labels to queries that are not given to $\cA$ before and also not corresponding to $\bot$. 
    We let them $\ell_1=L(x_1),...,\ell_t=L(x_t)$ and $\vecx=(x_1,...,x_t)$.
    We must maintain the representations of the elements of $\cA_\ord$ as a polynomial in $\Z[X_1,...,X_t]$ where $X_i$ corresponds to $x_i$. The definition of informative collisions is a pair of elements that have the same labels but as the polynomials different, and their difference is not included in the span of the previous informative collisions. Note that the algorithm must find at least one informative collision. Let us assume that it is represented by
    \[
        P(x_1,...,X_t)=a_1 X_1 + ... + a_t X_t + c = 0,
    \]
    where we can assume that $|a_i|,|c| \le 2^T N$ by the same reason to the original proof.
    
    We make the encoding scheme for $(N,\vecx).$
    \begin{description}
        \item[$\Encode(N,\vecx)$:] It runs $\cA_\ord^{\cG_N}()$ and computes the first informative collision $c$.
        It additionally includes $\vecx$ as a part of the encoding. 
        Note that the probability that $M=P(x_1,...,x_t)=0\bmod p$ for another $n$-bit prime $p$ is $1/p$, and $|P(x_1,...,x_t)| \le (t+1) 2^T N^2.$ Let $\ell$ be the index of $N$ among the divisor of $M$.
        \item[$\Decode(c,\ell,\vecx)$:] If $c=\bot$, it outputs a random sample from $\cD_\pprime^{(n)}$. 
        Otherwise, it recovers the first informative collision and plugs $\vecx$ to $X_1,...,X_t$ to compute $M=P(x_1,...,x_t)$, and outputs the $\ell$-th prime factor $N'.$
    \end{description}
    The encoding size is $3\log T - \log\log N + \log \binom{|\cG|}{t}+ O(1)$, which should be larger than $\log \frac{2^n}{n}  + \log \binom{|\cG|}{t} + \log \epsilon$. Rearranging this concludes the proof.
    \ifnum\llncs=1 \qed \fi
\end{proof}

\end{document}